\documentclass[twoside,11pt]{article}

\usepackage{jmlr2e}
\usepackage{amsmath,amsfonts,amssymb}
\usepackage{graphicx}
\usepackage{algorithm}
\usepackage{algorithmic}
\usepackage[dvips]{color}

\newcommand{\Id}{\mathtt{I}}

\newcommand{\rank}[1]{\ensuremath{\mathrm{rank}\left(#1\right)}}
\newcommand{\norm}[1]{\ensuremath{\left\| #1 \right\|}}

\newtheorem{problem}{Problem}
\newtheorem{defn}{Definition}
\firstpageno{1}

\begin{document}

\title{Recovery of Coherent Data via Low-Rank Dictionary Pursuit}

\author{\name Guangcan Liu \email guangcan.liu@rutgers.edu\\
       \addr Department of Statistical Science,  Cornell University, Ithaca, NY 14853, USA\\
       Department of Statistics and Biostatistics, Rutgers University, Piscataway, NJ 08854, USA
       \AND
       \name Ping Li \email pingli@stat.rutgers.edu\\
       \addr Department of Statistics and Biostatistics, Department of Computer Science,\\
             Rutgers University, Piscataway, NJ 08854, USA}

\editor{***}

\maketitle
\begin{abstract}
The recently established RPCA~\citep{Candes:2009:JournalACM} method provides us a convenient way to restore low-rank matrices from grossly corrupted observations. While elegant in theory and powerful in reality, RPCA may be not an ultimate solution to the low-rank matrix recovery problem. Indeed, its performance may not be perfect even when the data is strictly low-rank. This is because RPCA prefers incoherent data, which, however, could be inconsistent with some natural structures of data. As a typical example, consider the clustering structure which is ubiquitous in modern applications. As the number of cluster grows, the coherence parameters of data keep increasing, and accordingly, the recovery performance of RPCA degrades. We show that it is possible for Low-Rank Representation (LRR)~\citep{tpami_2013_lrr} to overcome the challenges raised by coherent data, as long as the dictionary in LRR is configured appropriately. Namely, we mathematically prove that if the dictionary itself is low-rank then LRR can avoid the coherence parameters which have potential to be large. This provides an elementary principle for dealing with coherent data and naturally leads to a practical algorithm for obtaining proper dictionaries in unsupervised environments. Our extensive experiments on randomly generated matrices and real motion sequences show promising results.
\end{abstract}
\section{Introduction}
Nowadays our data is often high-dimensional, massive and full of gross errors (e.g., corruptions, outliers and missing measurements). In the presence of gross errors, the classical Principal Component Analysis (PCA) method, which is probably the most widely used tool for data analysis and dimensionality reduction, becomes brittle --- A single gross error could render the estimate produced by PCA arbitrarily far from the desired estimate. As a consequence, it is crucial to develop new statistical tools for robustifying PCA. A variety of methods have been proposed and explored in the literature over several decades, e.g.,~\citep{CandesPIEEE,Candes:2009:math,Candes:2009:JournalACM,cacm_1981_ransac,robust:1972,gross:2011:tit,ke:cvpr:2005,torre:ijcv:2003,xu:2013:tit,tpami_2013_lrr,rahul:jlmr:2010,robust_spc,xu:2010:nips}. One of the most exciting methods is probably the so-called RPCA (Robust Principal Component Analysis) method by~\citet{Candes:2009:JournalACM}, built upon the exploration of the following low-rank matrix recovery problem:
\begin{problem}[Low-Rank Matrix Recovery]\label{pb:lmr}
Suppose we have a data matrix $X\in\mathbb{R}^{m\times{}n}$ and we know it can be decomposed as
\begin{eqnarray}\label{eq:recover}
X = L_0 + S_0,
\end{eqnarray}
where $L_0\in\mathbb{R}^{m\times{}n}$ is a low-rank matrix in which each column is a data point drawn from some low-dimensional subspace, and $S_0\in\mathbb{R}^{m\times{}n}$ is a sparse matrix supported on $\Omega\subseteq{}\{1,\cdots,m\}\times\{1,\cdots,n\}$. Except these mild restrictions, both components are arbitrary. The rank of $L_0$ is unknown, the support set $\Omega$ (i.e., the locations of the nonzero entries of $S_0$) and its cardinality (i.e., the amount of the nonzero entries of $S_0$) are unknown either. In particular, the magnitudes of the nonzero entries in $S_0$ may be arbitrarily large. Given $X$, can we recover both $L_0$ and $S_0$, in a scalable and exact fashion?
\end{problem}

The theory of RPCA tells us that, very generally, when the low-rank matrix $L_0$ satisfies some \emph{incoherent conditions} (i.e., the coherence parameters of $L_0$ are small), both the low-rank and the sparse matrices can be exactly recovered by using the following convex, potentially scalable program:
\begin{eqnarray}\label{eq:rpca}
\min_{L,S}\|L\|_* + \lambda \|S\|_1, &\textrm{s.t.}&X = L+S,
\end{eqnarray}
where $\|\cdot\|_*$ is the nuclear norm~\citep{phd_2002_nuclear} of a matrix, $\|\cdot\|_1$ denotes the $\ell_1$ norm of a matrix seen as a long vector, and $\lambda>0$ is a parameter. Besides of its elegance in theory, RPCA also has good empirical performance in many practical areas, e.g., image processing~\citep{zhang:2012:ijcv}, computer vision~\citep{peng:pami:2012}, radar imaging~\citep{sar:2012}, magnetic resonance imaging~\citep{mri:2012}, etc.

While complete in theory and powerful in reality, RPCA cannot be an ultimate solution to the low-rank matrix recovery Problem~\ref{pb:lmr}. Indeed, the method might not produce perfect recovery even when the latent matrix $L_0$ is strictly low-rank. This is because, seen from the aspect of mathematics, RPCA requires $L_0$ to satisfy some incoherent conditions, which, however, might not hold in reality. In a physical sense, the reason is that RPCA captures only the low-rankness property, which should not be the only property of our data, but essentially ignores the \emph{extra structures} (beyond low-rankness) widely existed in data: Given the situation that $L_0$ is low-rank, i.e., the column vectors of $L_0$ locate on a low-dimensional subspace, it is quite normal that $L_0$ may have some extra structures, which specify in more detail \emph{how} the data points (i.e., the column vectors of $L_0$) locate on the subspace.

Figure~\ref{fig:cluster} demonstrates a typical example of extra structures; that is, the clustering structure which is ubiquitous in modern applications~\citep{ijcv_1998_factor,cvpr_2009_ssc,robust_spc}. Whenever the data is exhibiting some clustering structure, the coherence parameters might be large and therefore RPCA might be unsatisfactory. More precisely, as will be shown in this paper, while the rank of $L_0$ is fixed and the underlying cluster number goes large, the coherence of $L_0$ keeps heightening and thus, arguably, the performance of RPCA drops.

\newpage

\begin{figure}[h!]
\begin{center}
\includegraphics[width=0.85\textwidth]{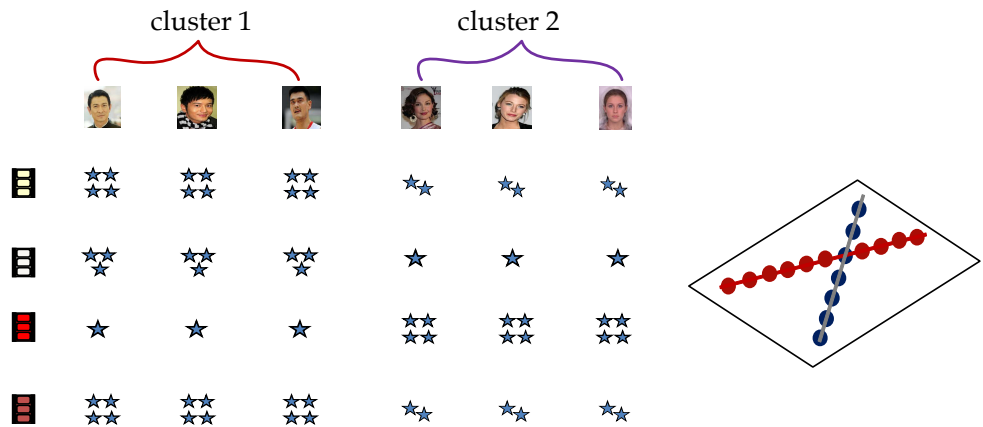}
\caption{Exemplifying the extra structures of low-rank data. Consider the Netflix data where each entry of the data matrix is a grade that a user assigns to a movie. Usually, such data is low-rank, as there exist wide correlations among the grades that different users assign to the same movie. Also, such data could own some clustering structure, since the preferences of the same kind of users are more similar to each other than to those with different gender, personality, culture and education background. In summary, such data (1) is usually low-rank and (2) often exhibits some clustering structure.}\label{fig:cluster}
\end{center}\vspace{-0.2in}
\end{figure}

To well handle \emph{coherent data}\footnote{Generally, coherent (resp. incoherent) data refers to the matrices whose coherence parameters are relatively large (resp. small). Yet there is no deterministic threshold to divide all matrices into coherent matrices and incoherent ones. To avoid confusion, in this paper we say that a matrix is incoherent if and only if the column vectors of the matrix are sampled from a single subspace uniformly at random. Apart from this particular case, the matrix is said to be coherent. In that sense, strictly speaking, the ``incoherent data'' stated in this paper does not exist in realistic environments.}, a straightforward approach is to \emph{avoid} the coherence parameters of $L_0$. Nevertheless, as explained in~\citep{Candes:2009:JournalACM,Candes:2009:math}, the coherence parameters are indeed \emph{necessary} for matrix recovery (if there is no additional condition available). Even more, this paper shall further indicate that the coherence parameters are related in nature to some extra structures intrinsically existed in $L_0$ and therefore \emph{cannot} be discarded simply. Interestingly, we show that it is possible to \emph{avoid} the coherence parameters by imposing some \emph{additional conditions}, which are easy to obey in supervised environments and can also be approximately satisfied in unsupervised environments. Our study is based on the following convex program termed Low-Rank Representation (LRR)~\citep{tpami_2013_lrr}:
\begin{eqnarray}\label{eq:lrr}
\min_{Z,S}\|Z\|_* + \lambda \|S\|_1, &\textrm{s.t.}&X = AZ+S,
\end{eqnarray}
where $A\in\mathbb{R}^{m\times{}d}$ is a size-$d$ dictionary matrix constructed in advance\footnote{Note that it is unimportant to determine the value of $d$. Suppose $Z^*$ is the optimal solution with respect to $Z$. Then LRR uses $AZ^*$ to restore $L_0$. It is easy to see that LRR falls back to RPCA when $A=\Id$ (identity matrix), and it can actually be further proved that the recovery produced by LRR is the same as RPCA whenever the dictionary $A$ is orthogonal.}, and $\lambda>0$ is a parameter. In order for LRR to avoid the coherence parameters which have potential to be large in the presence of extra structures, we prove that it is sufficient to construct in advance a dictionary matrix $A$ which is low-rank by itself. This additional condition (i.e., the dictionary $A$ is low-rank) gives a generic prescription to defend the possible infections raised by coherent data, providing an elementary criterion for learning the dictionary matrix $A$. Subsequently, we propose a simple and effective algorithm that utilizes the output of RPCA to construct the dictionary in LRR. Our extensive experiments demonstrated on randomly generated matrices and motion data show promising results. In summary, the contributions of this paper include:
\begin{itemize}
\item[$\diamond$] For the first time, this paper studies the problem of recovering low-rank, but coherent matrices from their corrupted versions. We investigate the physical regime where coherent data arises --- The widely existed clustering structure is a typical example that leads to coherent data. We prove some basic theories for resolving the problem of recovering coherent data, and also establish a practical algorithm that works better than RPCA in our experiments.
\item[$\diamond$] The studies of this paper help to understand the \emph{physical} meaning of coherence, which is now standard and widely used in various literatures, e.g.,~\citep{CandesPIEEE,Candes:2009:math,Candes:2009:JournalACM,xu:2010:nips,jmlr_2012_lrr}. We show that the coherence parameters are not ``assumptions'' for accomplishing a proof, but rather some excellent quantities that relate in nature to the \emph{extra structures} (beyond low-rankness) intrinsically existed in $L_0$.
\item[$\diamond$] This paper provides insights regarding the LRR model proposed by~\citep{tpami_2013_lrr}. While the special case of $A=X$ has been extensively studied, the LRR model \eqref{eq:lrr} with general dictionaries was not fully understood. We show that LRR \eqref{eq:lrr} equipped with proper dictionaries could well handle coherent data.
\item[$\diamond$] The idea of replacing $L$ with $AZ$ is essentially related to the spirit of matrix factorization which has been explored for long, e.g.,~\citep{Srebro05generalizationerror,nips:WeimerKLS07}. In that sense, the explorations of this paper help to understand why factorization techniques are useful.
\end{itemize}

The remainder of this paper is organized as follows. Section~\ref{sec:notation} summarizes mathematical notations used throughout this paper. In Section~\ref{sec:principle}, we explore the problem of recovering coherent data from corrupted observations, providing some theories and an algorithm for resolving the problem. Section~\ref{sec:proof} presents the complete proof procedure of our main result. Section~\ref{sec:exp} demonstrates experimental results and Section \ref{sec:con} concludes this paper.
\section{Summary of Main Notations}\label{sec:notation}
Capital letters such as $M$ are used to represent matrices, and accordingly, $[M]_{ij}$ denotes its $(i,j)$th entry. Letters $U$, $V$, $\Omega$ and their variants (complements, subscripts, etc.) are reserved for left singular vectors, right singular vectors and support set, respectively. We slightly abuse the notation $U$ (resp. $V$) to denote the linear space spanned by the columns of $U$ (resp. $V$), i.e., the column space (resp. row space). The projection onto the column space $U$, is denoted by $\mathcal{P}_U$ and given by $\mathcal{P}_U(M)=UU^TM$, and similarly for the row space $\mathcal{P}_V(M)=MVV^T$. We  also abuse the notation $\Omega$ to denote the linear space of matrices supported on $\Omega$. Then $\mathcal{P}_{\Omega}$ and $\mathcal{P}_{\Omega^{\bot}}$ respectively denote the projections onto $\Omega$ and $\Omega^c$ such that $\mathcal{P}_{\Omega}+\mathcal{P}_{\Omega^{\bot}}=\mathcal{I}$, where $\mathcal{I}$ is the identity operator. The symbol $(\cdot)^+$  denotes the Moore-Penrose pseudoinverse of a matrix: $M^+=V_M\Sigma_M^{-1}U_M^T$ for a matrix $M$ with Singular Value Decomposition (SVD)\footnote{In this paper, SVD always refers to skinny SVD. For a rank-$r$ matrix $M\in\mathbb{R}^{m\times{}n}$, its SVD is of the form $U_M\Sigma_MV_M^T$, with $U_M\in\mathbb{R}^{m\times{}r},\Sigma_M\in\mathbb{R}^{r\times{}r}$ and $V_M\in\mathbb{R}^{n\times{}r}$.} $U_M\Sigma_MV_M^T$.

Six different matrix norms are used in this paper. The first three norms are functions of the singular values: 1) The operator norm (i.e., the largest singular value) denoted by $\|M\|$, 2) the Frobenius norm (i.e., square root of the sum of squared singular values) denoted by $\|M\|_F$, and 3) the nuclear norm (i.e., the sum of singular values) denoted by $\|M\|_*$. The other three  are the $\ell_1$, $\ell_{\infty}$ (i.e., sup-norm) and $\ell_{2,\infty}$ norms of a matrix:  $\norm{M}_{1}=\sum_{i,j}|[M]_{ij}|$, $\norm{M}_{\infty}=\max_{i,j}\{|[M]_{ij}|\}$ and $\norm{M}_{2,\infty}=\max_{j}\{\sqrt{\sum_{i}[M]_{ij}^2}\}$.

The Greek letter $\mu$ and its variants (e.g., subscripts and superscripts) are reserved to denote the coherence parameters of a matrix. We shall also reserve two lower case letters, $m$ and $n$, to respectively denote the data dimension and the number of data points, and we use the following two symbols throughout this paper:
\begin{eqnarray*}
n_1 = \max(m,n) &\textrm{and}&n_2 = \min(m,n).
\end{eqnarray*}
A complete list of notations can be found in Appendix~\ref{sec:app:notations} for convenience of readers.
\section{On the Recovery of Coherent Data}\label{sec:principle}
In this section, we shall firstly investigate the physical regime that raises coherent data, and then discuss the problem of recovering coherent data from corrupted observations, providing some basic principles and an algorithm for resolving the problem.
\subsection{Coherence Parameters and Their Properties}
Notice that the rank function cannot fully capture all characteristics of $L_0$, and thus it is indeed necessary to define some quantities for measuring the effects of various extra structures (beyond low-rankness) such as the clustering structure demonstrated in Figure~\ref{fig:cluster}. The \emph{coherence} parameters defined in~\citep{Candes:2009:math,Candes:2009:JournalACM} are excellent exemplars of such quantities.
\subsubsection{$\mu_1$ and $\mu_2$}
For an $m\times{}n$ matrix $L_0$ with rank $r_0$ and SVD $L_0=U_0\Sigma_0V_0^T$, some of its important properties can be characterized by two coherence parameters, denoted as $\mu_1$ and $\mu_2$. The first coherence parameter, $1\leq\mu_1\leq{}m$, which characterizes the column space identified by $U_0$, is defined as
\begin{eqnarray}\label{eq:u1}
\mu_1(L_0) = \frac{m}{r_0}\max_{1\leq i\leq m}\|U_0^Te_i\|_2^2,
\end{eqnarray}
where $e_i$ denotes the $i$th standard basis. The second coherence parameter, $1\leq\mu_2\leq{}n$, which characterizes the row space identified by $V_0$, is defined as
\begin{eqnarray}\label{eq:u2}
\mu_2(L_0) = \frac{n}{r_0}\max_{1\leq j\leq n}\|V_0^Te_j\|_2^2.
\end{eqnarray}
In~\citep{Candes:2009:JournalACM}, another coherence parameter, called as the third coherence parameter and denoted as $1\leq\mu_3\leq{}mn$, is also introduced:
\begin{eqnarray*}
\mu_3(L_0) =  \frac{mn}{r_0}(\|U_0V_0^T\|_{\infty})^2=\frac{mn}{r_0}\max_{i,j}(|\langle{}U_0^Te_i,V_0^Te_j\rangle|)^2.
\end{eqnarray*}
Notice that $\mu_3$ is not indispensable, as it is actually a ``derivative'' of $\mu_1$ and $\mu_2$: Simple calculations give that $\mu_3\leq{}r_0\mu_1\mu_2$. The analysis of work does not need to access $\mu_3$. We include it just for the sake of consistence with~\citep{Candes:2009:JournalACM}.

The analysis in~\citep{Candes:2009:JournalACM} merges the above three parameters into a single one: $\mu(L_0)=\max\{\mu_1(L_0),\mu_2(L_0),\mu_3(L_0)\}$. As will be seen later, the behaviors of those three coherence parameters are different from each other, and thus it is indeed more adequate to consider them individually.
\subsubsection{$\mu_2$-phenomenon}\label{sec:u2}
\citet{Candes:2009:JournalACM} have proven that the success condition (regarding $L_0$) of RPCA is
\begin{eqnarray}\label{rpca:cond}
\rank{L_0}\leq{}\frac{n_2}{c_r\mu(L_0)(\log{n_1})^2},
\end{eqnarray}
where $\mu(L_0)=\max\{\mu_1(L_0),\mu_2(L_0),\mu_3(L_0)\}$ and $c_r>1$ is some numerical constant. So, RPCA will be less successful when the coherence parameters are considerably larger: The success condition \eqref{rpca:cond} is narrowed when $\mu(L_0)$ goes large. As an extreme example, consider the case where the latent matrix $L_0$ is one in only one column and zero everywhere else. Such a matrix produces $\mu_2(L_0)=n\geq{}n_2$, and thus the success condition \eqref{rpca:cond} is invalid. In this subsection, we shall further show that the widely existed clustering structure can enlarge the coherence parameters and, accordingly, degrades the performance of RPCA.
\begin{figure}[h!]
\begin{center}
\includegraphics[width=0.95\textwidth]{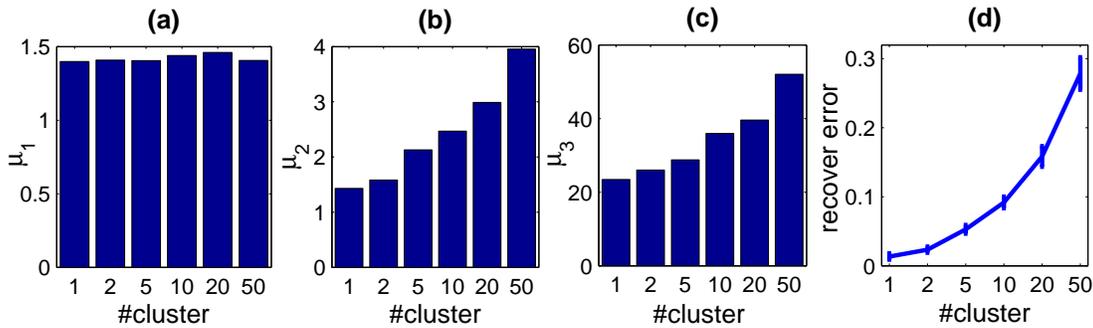}\vspace{-0.1in}
\caption{Exploring the influence of the cluster number, using randomly generated matrices. The size and rank of $L_0$ are fixed to be $500\times500$ and 100, respectively. The underlying cluster number is varying from 1 to 50. $S_0$ is fixed as a sparse matrix with 13\% nonzero entries. (a) The first coherence parameter $\mu_1(L_0)$ vs cluster number. (b) $\mu_2(L_0)$ vs cluster number. (c) $\mu_3(L_0)$ vs cluster number. (d) Recover error (produced by RPCA) vs cluster number. The numbers shown in above figures are averaged from 100 random trials. The recover error is computed as $\|\hat{L}_0-L_0\|_F/\|L_0\|_F$, where $\hat{L}_0$ denotes an estimate of $L_0$.}\label{fig:ic}\vspace{-0.2in}
\end{center}
\end{figure}

Given the situation that $L_0$ is low-rank, i.e., $\rank{L_0}\equiv{}r_0\ll{}n_2$, the data points (i.e., column vectors of $L_0$) should be sampled from a $r_0$-dimensional subspace. Yet the sampling is unnecessary to be \emph{uniform}. Indeed, a more realistic interpretation is to consider the data points as samples from the union of $k$ number of subspaces (i.e., clusters), and the sum of those multiple subspaces together has a dimension $r_0$. That is to say, there are multiple ``small'' subspaces inside one $r_0$-dimensional ``large'' subspace, as exemplified in Figure~\ref{fig:cluster}. It is arguable that such a structure of multiple subspaces exists widely in various domains, e.g., face, texture and motion~\citep{ijcv_1998_factor,cvpr_2009_ssc,Liu:2010:HGM}. Whenever the low-rank matrix $L_0$ is exhibiting such clustering behaviors, the second coherence parameter $\mu_2(L_0)$ will increase with the cluster number underlying $L_0$, as shown in Figure~\ref{fig:ic}. When the coherence is heightening, \eqref{rpca:cond} suggests that the performance of RPCA will drop, as verified in Figure~\ref{fig:ic}(d). For the ease of citation, we call the phenomena shown in Figure~\ref{fig:ic}(b)$\sim$(d) as the ``$\mu_2$-phenomenon''.

To see why the second coherence parameter increases with the cluster number underlying $L_0$, please refer to Appendix~\ref{app:why}. As can be seen from Figure~\ref{fig:ic}(a), the first coherence parameter $\mu_1$ is \emph{invariant} to the variation of the clustering number. This is because the behaviors of the data points (i.e., column vectors) can only affect the row space, while $\mu_1$ is defined on the column space. Nevertheless, if the row vectors of $L_0$ also own some clustering structure, $\mu_1$ could be large as well. This kind of data exists widely in text documents and we leave it as future work.
\subsection{Avoiding $\mu_2$ by LRR}
To accurately recover coherent matrices from their corrupted versions, one may establish some parametric models to \emph{capture} the extra structures which produce high coherence. However, it is usually hard, if not impossible, to know in advance what kind of extra structures there are and which models are appropriate to use. Even if the modalities of the extra structure are known, e.g., the mixture of multiple subspaces shown in Figure~\ref{fig:cluster}, such a strategy still needs to face some difficult problems, e.g., the estimate of the cluster number. In sharp contrast, it is much simpler to devise an approach that can \emph{avoid} the second coherence parameter $\mu_2$. Unfortunately, as explained in~\citep{Candes:2009:math,Candes:2009:JournalACM,jmlr_2012_lrr}, the coherence parameters are \emph{necessary} for identifying accurately the success conditions of matrix recovery. Even more, the $\mu_2$-phenomenon actually implies that $\mu_2$ is related in nature to some intrinsic structures of $L_0$ and thus cannot be eschewed freely. Interestingly, we shall show that LRR can avoid $\mu_2$ by using some \emph{additional conditions}, which are possible to obey in both supervised and unsupervised environments.\\

\noindent\textbf{Main Result:} We shall show that, when the dictionary matrix $A$ itself is low-rank, the recovery performance of LRR does not depend on $\mu_2$. Our main result is presented in the following theorem (The detailed proof procedure is deferred until Section~\ref{sec:proof}).
\begin{theorem}[Noiseless]\label{thm:noiseless}
Let $A\in\mathbb{R}^{m\times{}d}$ with SVD $A=U_A\Sigma_AV_A^T$ be a column-wisely unit-normed (i.e., $\|Ae_i\|_2=1,\forall{}i$) dictionary matrix which satisfies $\mathcal{P}_{U_A}(U_0)=U_0$ (i.e., $U_0$ is a subspace of $U_A$). For any $0<\epsilon<0.5$ and some numerical constant $c_a>1$, if
\begin{eqnarray}\label{eq:succ}
\rank{L_0}\leq\rank{A}\leq\frac{\epsilon^2n_2}{c_a\mu_1(A)\log{}n_1}&\textrm{and}&|\Omega|\leq(0.5-\epsilon)mn,
\end{eqnarray}
then with probability at least $1-n_1^{-10}$, the optimal solution to the LRR problem \eqref{eq:lrr} with $\lambda=1/\sqrt{n_1}$ is unique and exact, in a sense that
\begin{eqnarray*}
Z^* = A^+L_0&\textrm{and}& S^*=S_0,
\end{eqnarray*}
where $(Z^*,S^*)$ is the optimal solution to \eqref{eq:lrr}.
\end{theorem}

By $U_0\subset{}U_A$, the column space of $A$ should approximately have the same properties as $L_0$, and thus, roughly, $\mu_1(A)\approx\mu_1(L_0)$. So, as aforementioned, this paper needs to assume that the first coherence parameter of $L_0$ is small and only addresses the cases where the second coherence parameter might be large. It is worth noting that the restriction $\rank{L_0}\leq{}O(n_2/\log{}n_1)$ is looser than that of PRCA\footnote{In terms of \emph{exact} recovery, $O(n_2/\log{}n_1)$ is probably the ``finest'' bound one could accomplish in theory.}, which requires $\rank{L_0}\leq{}O(n_2/(\log{}n_1)^2)$. The requirement of column-wisely unit-normed dictionary (i.e., $\|Ae_i\|_2=1,\forall{}i$) is purely for complying the parameter estimate of $\lambda=1/\sqrt{n_1}$, which is consistent with RPCA. The condition $\mathcal{P}_{U_A}(U_0)=U_0$, i.e., $U_0$ is a subspace of $U_A$, is indispensable if we ask for exact recovery, because $\mathcal{P}_{U_A}(U_0)=U_0$ is implied by the equality $AZ^*=L_0$. This necessary condition, together with the condition that $A$ is low-rank, indeed provides an elementary criterion for learning the dictionary matrix $A$ in LRR. Figure~\ref{fig:demo} presents an example, which further confirms our main result: LRR is able to avoid $\mu_2$ as long as $U_0\subset{}U_A$ and $A$ is low-rank. Note that it is unnecessary for the dictionary $A$ to strictly satisfy $U_A=U_0$, and LRR is actually tolerant to the ``errors'' possibly existing in the dictionary.
\begin{figure}[h!]
\begin{center}\vspace{0.1in}
\includegraphics[width=0.95\textwidth]{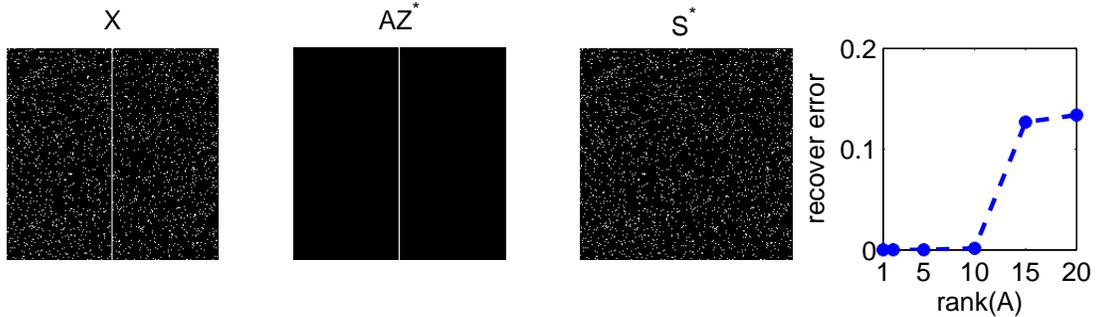}\vspace{-0.1in}
\caption{Exemplifying that LRR can void $\mu_2$. In this experiment, $L_0$ is a $200\times200$ rank-1 matrix with one column being $\mathbf{1}$ (i.e., a vector of all ones) and everything else being zero. Thus, $\mu_1(L_0)=1$ and $\mu_2(L_0)=n=200$. The sparse matrix $S_0$ is with Bernoulli $\{0,1\}$ values, and its nonzero fraction is set as 5\%. The dictionary is set as $A=[\mathbf{1},W]$ ($A$ is further normalized), where $W$ is a $200\times{}p$ random Gaussian matrix ($p$ is varying). As long as $\rank{A}=p+1\leq10$, LRR with $\lambda=0.08$ can exactly recover $L_0$ from a grossly corrupted observation matrix $X$.}\label{fig:demo}\vspace{-0.2in}
\end{center}
\end{figure}

The LRR program \eqref{eq:lrr} is designed for the cases where the uncorrupted observations are noiseless. In reality this is often not true and all entries of $X$ can be contaminated by a small amount of noises, i.e., $X=L_0+S_0+N$, where $N$ is a matrix of dense Gaussian noises. In this case, the formula of LRR \eqref{eq:lrr} need be modified to
\begin{eqnarray}\label{eq:lrr:noisy}
\min_{Z,S}\|Z\|_*+\lambda\|S\|_1, &\textrm{s.t.}&\|X-AZ-S\|_F\leq{}\varepsilon,
\end{eqnarray}
where $\varepsilon$ is a parameter that measures the noise level of data. In the experiments of this paper, we consistently use $\varepsilon=10^{-6}\|X\|_F$. In the presence of dense noises, the latent matrices, $L_0$ and $S_0$, cannot be exactly restored. Yet we have the following theorem to guarantee the near recovery property of the solution produced by \eqref{eq:lrr:noisy} (please refer to Appendix~\ref{app:proof:noisy} for the proof):
\begin{theorem}[Noisy]\label{thm:noisy}
Suppose $\|X-L_0-S_0\|_F\leq\varepsilon$. Let $A\in\mathbb{R}^{m\times{}d}$ with SVD $A=U_A\Sigma_AV_A^T$ be a column-wisely unit-normed dictionary matrix which satisfies $\mathcal{P}_{U_A}(U_0)=U_0$. For any $0<\epsilon<0.35$ and some numerical constant $c_a>1$, if
\begin{eqnarray}\label{eq:succ:noisy}
\rank{L_0}\leq\rank{A}\leq\frac{\epsilon^2n_2}{c_a\mu_1(A)\log{}n_1}&\textrm{and}&|\Omega|\leq(0.35-\epsilon)mn,
\end{eqnarray}
then with probability at least $1-n_1^{-10}$, any solution $(Z^*,S^*)$ to the LRR program \eqref{eq:lrr:noisy} with $\lambda=1/\sqrt{n_1}$ gives a near recovery to $(L_0,S_0)$, in a sense that $\|AZ^*-L_0\|_F\leq{}8\sqrt{mn}\varepsilon$ and $\|S^*-S_0\|_F\leq(8\sqrt{mn}+2)\varepsilon$.
\end{theorem}
\subsection{An Unsupervised Algorithm for Matrix Recovery}
To well handle coherent data, Theorem~\ref{thm:noiseless} suggests that, ideally, the dictionary matrix $A$ should be low-rank and satisfy $U_0\subset{}U_A$. In certain supervised environment, this would be easy as one could use clear, well-processed training data to construct the dictionary. In unsupervised environments, however, it is challenging to purse a low-rank dictionary that can also satisfy $U_0\subset{}U_A$, since $U_0\subset{}U_A$ is essentially some kind of ``weak'' supervision information: As long as the dictionary matrix $A$ is low-rank, $U_0\subset{}U_A$ forms a prior that $L_0$ is known to be contained by a low-rank subspace identified by $U_A$. Interestingly, as will be shown later, it is possible to approximate the desired dictionary even when no prior about $L_0$ is given.

We shall introduce a heuristic algorithm that works distinctly better than RPCA in our experiments. As can be seen from \eqref{rpca:cond}, RPCA is actually not brittle with respect to coherent data: Except for the extreme case where the coherence parameters reach the upper bound $n$ (or $m$), RPCA could own a valid condition (although the condition is narrowed) to be \emph{exactly} successful even when the coherence parameters are considerably large. Based on this, we propose a pretty simple algorithm, as summarized in Algorithm \ref{alg:mr}, to achieve a \emph{solid} improvement over RPCA. Our idea is straightforward: We firstly obtain an estimate of $L_0$ by using RPCA and then utilize the estimate to construct the dictionary matrix $A$. The post-processing steps (Step 2 and Step 3) that slightly modify the solution of RPCA are designed to encourage well-conditioned dictionary, which is the favorite circumstance indicated by Theorem~\ref{thm:noiseless}.
\begin{algorithm}[h!]
\caption{Matrix Recovery}\label{alg:mr}
\begin{algorithmic}
\STATE \textbf{input:} Observed data matrix $X\in{}\mathbb{R}^{m\times{}n}$.
\STATE\textbf{adjustable parameter:} $\lambda$.
\STATE \textbf{1.} Solve for $\hat{L}_0$ by optimizing the RPCA problem \eqref{eq:rpca} with $\lambda=1/\sqrt{n_1}$.
\STATE \textbf{2.} Estimate the rank of $\hat{L}_0$ by $$\hat{r}_0=\#\{i: \sigma_i>10^{-3}\sigma_1\},$$
where $\sigma_1,\sigma_2,\cdots,\sigma_{n_2}$ are the singular values of $\hat{L}_0$.
\STATE \textbf{3.} Form $\tilde{L}_0$ by using the rank-$\hat{r}_0$ approximation of $\hat{L}_0$. That is,
\begin{eqnarray*}
\tilde{L}_0=\arg\min_{L}\|L-\hat{L}_0\|_F^2, \textrm{ s.t. } \rank{L}\leq\hat{r}_0,
\end{eqnarray*}
which is solved by SVD.
\STATE \textbf{4.} Construct a dictionary $\hat{A}$ from $\tilde{L}_0$ by normalizing the column vectors of $\tilde{L}_0$:
\begin{eqnarray*}
[\hat{A}]_{:,i}=\frac{[\tilde{L}_0]_{:,i}}{\|[\tilde{L}_0]_{:,i}\|_2}, i=1,\cdots,n,
\end{eqnarray*}
where $[\cdot]_{:,i}$ denotes the $i$th column of a matrix.
\STATE\textbf{5.} Solve for $Z^*$ by optimizing the LRR problem \eqref{eq:lrr} with $A=\hat{A}$ and $\lambda=1/\sqrt{n_1}$.
\STATE \textbf{output:} $\hat{A}Z^*$.
\end{algorithmic}
\end{algorithm}

Whenever the recovery produced by RPCA is already exact, the claim in Theorem~\ref{thm:noiseless} gives that the recovery produced by our Algorithm \ref{alg:mr} is exact as well. When RPCA fails to exactly recover $L_0$, the produced dictionary is still possible to satisfy the success conditions required by Theorem~\ref{thm:noiseless}, namely $A$ is low-rank and $U_0\subset{}U_A$. This is because those conditions are weaker than $A=L_0$. Thus, in terms of exactly recovering $L_0$ from a given $X$, the success probability of our Algorithm \ref{alg:mr} is greater than or equal to that of RPCA. Also, in a computational sense, Algorithm \ref{alg:mr} does not double RPCA, although there are two convex programs in our algorithm. In fact, according to our simulations, usually the computational time of Algorithm \ref{alg:mr} is just 1.2 times as much as RPCA. The reason is that, as has been explored by~\citep{tpami_2013_lrr}, the complexity of solving the LRR problem \eqref{eq:lrr} is $O(n^2r_A)$ (assume $m=n$), which is much lower than that of RPCA (which requires $O(n^3)$) provided that the obtained dictionary matrix $A$ is fairly low-rank (i.e., $r_A$ is small).

One may have noticed that the procedure of Algorithm \ref{alg:mr} could be made iterative, i.e., one can consider $\hat{A}Z^*$ as a new estimate of $L_0$ and use it to further update the dictionary matrix $A$, and so on. Nevertheless, we empirically find that such an iterative procedure often converges within two iterations. Hence, for the sake of simplicity, we do not consider the iterative strategies in this paper.
\section{Proof of Theorem~\ref{thm:noiseless}}\label{sec:proof}
\subsection{Settings and Some Basic Lemmas}
The same as in RPCA~\citep{Candes:2009:JournalACM}, we assume that the locations of the corrupted entries are selected \emph{uniformly at random}. In more details, we work with the Bernoulli model $\Omega=\{(i,j):\delta_{ij}=1\}$, where $\delta_{ij}$'s are i.i.d. variables taking value one with probability $\rho_0=|\Omega|/(mn)$ and zero with probability $(1-\rho_0)$, so that the expected cardinality of $\Omega$ is $\rho_0mn$. For the ease of presentation, we assume that the signs of the nonzero entries of $S_0$ are symmetric Bernoulli $\pm1$ values:
\begin{eqnarray*}
[sign(S_0)]_{ij}=\left\{\begin{array}{ll}
1,& \textrm{with probability } \frac{\rho_0}{2},\\
0,& \textrm{with probability } 1-\rho_0,\\
-1,&\textrm{with probability } \frac{\rho_0}{2}.
\end{array}\right.
\end{eqnarray*}
For general sign matrices, the same as in RPCA~\citep{Candes:2009:JournalACM}, our Theorem~\ref{thm:noiseless} can still be proved by globally placing an elimination theorem and a derandomization scheme. Yet the success conditions in Theorem~\ref{thm:noisy} have not been proven when $sign(S_0)$ has an arbitrary distribution, because the elimination theorem does not hold in the noisy case.

The following two lemmas are well-known and will be used multiple times in the proof.
\begin{lemma}\label{app:lem:basic:1}For any matrix $M$, the following holds:
\begin{itemize}
\item[1.] Let the SVD of $M$ be $U_M\Sigma_MV_M^T$. Then we have $\partial\|M\|_*=\{U_MV_M^T+W|U_M^TW=0, WV_M=0,\|W\|\leq1\}$.
\item[2.] Let the support set of $M$ be $\Omega_M$. Then we have $\partial\|M\|_{1}=\{sign(M)+F|\mathcal{P}_{\Omega_M}(F)=0, \|F\|_{\infty}\leq1\}$.
\end{itemize}
\end{lemma}
\begin{lemma}\label{lem:basic:2}
For any matrices $M$ and $N$ of consistent sizes,
\begin{align}\notag
&|\langle{}M,N\rangle|\leq\|M\|_{\infty}\|N\|_{1},\\\notag
&|\langle{}M,N\rangle|\leq\|M\|_F\|N\|_F,\\\notag
&\|MN\|_F\leq\|M\|\|N\|_F,\\\notag
&\|MN\|_{2,\infty}\leq\|M\||N\|_{2,\infty}.
\end{align}
\end{lemma}
\subsection{Critical Lemmas}
First of all, we would like to prove that the sparse matrix $S_0$ does not locate in the column space of the dictionary $A$, i.e., $U_A\cap\Omega=\{0\}$ or $\|\mathcal{P}_{U_A}\mathcal{P}_{\Omega}\|<1$ as equal. Provided that $A\in\mathbb{R}^{m\times{}d}$ is fairly low-rank, the analysis in~\citep{Candes:2009:JournalACM} gives that $$\|\mathcal{P}_{T_A}\mathcal{P}_{\Omega}\|\leq\sqrt{\frac{|\Omega|}{mn}+\epsilon}$$ holds with high probability for any $\epsilon>0$, where $T_A$ denotes the linear space given by $\mathcal{P}_{U_A}+\mathcal{P}_{V_A}-\mathcal{P}_{U_A}\mathcal{P}_{V_A}$. Since $U_A\subset{}T_A$ and $\|\mathcal{P}_{U_A}\mathcal{P}_{\Omega}\|\leq{}\|\mathcal{P}_{T_A}\mathcal{P}_{\Omega}\|$, it is natural to anticipate that $\|\mathcal{P}_{U_A}\mathcal{P}_{\Omega}\|$ is smaller than 1 with high probability. The difference is that we only need the first coherence parameter $\mu_1$ to finish the proof. Following the techniques in~\citep{Candes:2009:JournalACM}, we have the following lemma to bound the operator norm of $\mathcal{P}_{U_A}\mathcal{P}_{\Omega}$.

\begin{lemma}\label{lem:papo}
Suppose $\Omega\sim{}Ber(\rho_0)$ with $\rho_0<1$. Then for any $\epsilon>0$,
\begin{eqnarray*}
\|\mathcal{P}_{U_A}\mathcal{P}_{\Omega}\|\leq\sqrt{\rho_0+\epsilon}
\end{eqnarray*}
holds with probability at least $1-n_1^{-10}$, provided that
\begin{eqnarray*}
\rank{A}\leq\frac{\epsilon^2n_2}{c_a\mu_1(A)\log{}n_1}.
\end{eqnarray*}
\end{lemma}
\begin{proof}
For any matrix $M$, we have
\begin{eqnarray*}
\mathcal{P}_{U_A}(M)=\sum_{i,j}\langle{}\mathcal{P}_{U_A}(M),e_ie_j^T\rangle{}e_ie_j^T,
\end{eqnarray*}
and so
\begin{eqnarray*}
\mathcal{P}_{\Omega^\bot}\mathcal{P}_{U_A}(M)=\sum_{i,j}(1-\delta_{ij})\langle{}\mathcal{P}_{U_A}(M),e_ie_j^T\rangle{}e_ie_j^T,
\end{eqnarray*}
which gives
\begin{eqnarray*}
\mathcal{P}_{U_A}\mathcal{P}_{\Omega^\bot}\mathcal{P}_{U_A}(M)&=&\sum_{i,j}(1-\delta_{ij})\langle{}\mathcal{P}_{U_A}(M),e_ie_j^T\rangle{}\mathcal{P}_{U_A}(e_ie_j^T)\\
&=&\sum_{i,j}(1-\delta_{ij})\langle{}M,\mathcal{P}_{U_A}(e_ie_j^T)\rangle{}\mathcal{P}_{U_A}(e_ie_j^T).
\end{eqnarray*}
Note that the Frobenius norm of a matrix is equivalent to the vector $\ell_2$ norm, while considering the matrix as a long vector. In that sense, we have
\begin{eqnarray*}
\mathcal{P}_{U_A}\mathcal{P}_{\Omega^\bot}\mathcal{P}_{U_A}=\sum_{i,j}(1-\delta_{ij})\mathcal{P}_{U_A}(e_ie_j^T)\otimes\mathcal{P}_{U_A}(e_ie_j^T).
\end{eqnarray*}
The definition of $\mu_1(A)$ gives
\begin{eqnarray*}
\|\mathcal{P}_{U_A}(e_ie_j^T)\|_F^2\leq{}\frac{\mu_1(A)r_A}{m}.
\end{eqnarray*}
Then by using the results in~\citep{Rudelson99randomvectors} and following the proof procedure of~\citep{Candes:2009:math}, we have that
\begin{eqnarray*}
\|(1-\rho_0)\mathcal{P}_{U_A}-\mathcal{P}_{U_A}\mathcal{P}_{\Omega^\bot}\mathcal{P}_{U_A}\|&\leq&(1-\rho_0)(\phi_1\sqrt{\frac{\mu_1(A)r_A\log{}n_1}{n_2}}+\phi_2\sqrt{\frac{\mu_1(A)\beta{}r_A\log{}n_1}{n_2}})\\
&\leq&\phi_1\sqrt{\frac{\mu_1(A)r_A\log{}n_1}{n_2}}+\phi_2\sqrt{\frac{\mu_1(A)\beta{}r_A\log{}n_1}{n_2}}
\end{eqnarray*}
holds with probability at least $1-n_1^{-\beta}$ for some numerical constants $\phi_1$ and $\phi_2$. For any $\epsilon>0$, setting $\beta=10$ and $c_a=(\phi_1+\sqrt{10}\phi_2)^2$ gives that
\begin{eqnarray*}
\|(1-\rho_0)\mathcal{P}_{U_A}-\mathcal{P}_{U_A}\mathcal{P}_{\Omega^\bot}\mathcal{P}_{U_A}\|\leq\epsilon
\end{eqnarray*}
holds with probability at least $1-n_1^{-10}$, provided that $r_A\leq{}\epsilon^2n_2/(c_a\mu_1(A)\log{}n_1)$.

By $\mathcal{P}_{U_A}\mathcal{P}_{\Omega}\mathcal{P}_{U_A}=-\rho_0\mathcal{P}_{U_A}-((1-\rho_0)\mathcal{P}_{U_A}-\mathcal{P}_{U_A}\mathcal{P}_{\Omega^\bot}\mathcal{P}_{U_A})$ and the triangle inequality,
\begin{eqnarray*}
\|\mathcal{P}_{U_A}\mathcal{P}_{\Omega}\mathcal{P}_{U_A}\|&\leq&\|\rho_0\mathcal{P}_{U_A}\|+\|(1-\rho_0)\mathcal{P}_{U_A}-\mathcal{P}_{U_A}\mathcal{P}_{\Omega^\bot}\mathcal{P}_{U_A}\|\\
&\leq&\rho_0+\epsilon=\frac{|\Omega|}{mn}+\epsilon.
\end{eqnarray*}
Finally, the fact $\|\mathcal{P}_{U_A}\mathcal{P}_{\Omega}\mathcal{P}_{U_A}\|=\|\mathcal{P}_{U_A}\mathcal{P}_{\Omega}\|^2$ completes the proof.
\end{proof}

While the above Lemma implies that $\|\mathcal{P}_{U_A}\mathcal{P}_{\Omega}(M)\|_F\leq{}(\rho_0+\epsilon)\|M\|_F$, we often need to bound the sup-norm of $\mathcal{P}_{U_A}\mathcal{P}_{\Omega}(M)$. The next lemma will show that, when the signs of the matrix entries are independent symmetric Bernoulli variables, the sup-norm could be arbitrarily small.
\begin{lemma}\label{lem:inf:pas0}
Suppose $\mathcal{P}$ is a symmetric linear projection with $\|\mathcal{P}\|\leq2$, and $\Psi\in\mathbb{R}^{m\times{}n}$ is a random sign matrix with i.i.d. entries distributed as
\begin{eqnarray*}
[\Psi]_{ij}=\left\{\begin{array}{ll}
1,& \textrm{with probability } \frac{1}{2},\\
-1,&\textrm{with probability } \frac{1}{2}.
\end{array}\right.
\end{eqnarray*}
For any $\epsilon>0$,
\begin{eqnarray*}
\|\mathcal{P}_{U_A}\mathcal{P}\mathcal{P}_{U_A}\mathcal{P}_{\Omega}(\Psi)\|_{\infty}\leq\epsilon
\end{eqnarray*}
holds with high probability as long as
\begin{eqnarray*}
\rank{A}\leq\frac{\epsilon^2n_2}{c_a\mu_1(A)\log{}n_1}.
\end{eqnarray*}
\end{lemma}
\begin{proof}
Let $\xi_{ij}=[\Psi]_{ij}$ and
\begin{eqnarray*}
Q&=&\mathcal{P}_{U_A}\mathcal{P}\mathcal{P}_{U_A}\mathcal{P}_{\Omega}(\Psi)\\
&=&\mathcal{P}_{U_A}\mathcal{P}\mathcal{P}_{U_A}(\sum_{i,j}\delta_{ij}\xi_{ij}e_ie_j^T)\\
&=&\sum_{i,j}\delta_{ij}\xi_{ij}\mathcal{P}_{U_A}\mathcal{P}\mathcal{P}_{U_A}(e_ie_j^T).\\
\end{eqnarray*}
Then it can be seen that each entry of $Q$ is a sum of independent random variables:
\begin{eqnarray*}
[Q]_{i_1j_1} &=& \sum_{i,j}y_{ij} \textrm{ with }\\
y_{ij} &=&\delta_{ij}\xi_{ij}\langle\mathcal{P}_{U_A}\mathcal{P}\mathcal{P}_{U_A}(e_ie_j^T),e_{i_1}e_{j_1}^T\rangle.
\end{eqnarray*}
Note here that the variables $\delta_{ij}$'s are fixed and the randomness comes from $\xi_{ij}$'s.

It is easy to see that $\mathbb{E}(y_{ij})=0$. We have
\begin{eqnarray*}
|y_{ij}-\mathbb{E}(y_{ij})|&=&|\delta_{ij}\xi_{ij}\langle\mathcal{P}_{U_A}\mathcal{P}\mathcal{P}_{U_A}(e_ie_j^T),e_{i_1}e_{j_1}^T\rangle|\\
&=&|\delta_{ij}\xi_{ij}\langle\mathcal{P}_{U_A}(e_ie_j^T),\mathcal{P}\mathcal{P}_{U_A}(e_{i_1}e_{j_1}^T)\rangle|\\
&\leq&\|\mathcal{P}_{U_A}(e_ie_j^T)\|_F\|\mathcal{P}\mathcal{P}_{U_A}(e_{i_1}e_{j_1}^T)\|_F\\
&\leq&\|\mathcal{P}_{U_A}(e_ie_j^T)\|_F\|\mathcal{P}\|\|\mathcal{P}_{U_A}(e_{i_1}e_{j_1}^T)\|_F\\
&\leq&\frac{2u_1(A)r_A}{m}.
\end{eqnarray*}
We also have
\begin{eqnarray*}
\sum_{i,j}Var(y_{ij}) &=&\sum_{ij}|\delta_{ij}\langle\mathcal{P}_{U_A}\mathcal{P}\mathcal{P}_{U_A}(e_ie_j^T),e_{i_1}e_{j_1}^T\rangle|^2Var(\xi_{ij})\\
&=&\sum_{i,j}(\delta_{ij})^2|\langle\mathcal{P}_{U_A}\mathcal{P}\mathcal{P}_{U_A}(e_ie_j^T),e_{i_1}e_{j_1}^T\rangle|^2\\
&=&\sum_{i,j}(\delta_{ij})^2|\langle{}e_ie_j^T,\mathcal{P}_{U_A}\mathcal{P}\mathcal{P}_{U_A}(e_{i_1}e_{j_1}^T)\rangle|^2\\
&\leq&\sum_{i,j}|\langle{}e_ie_j^T,\mathcal{P}_{U_A}\mathcal{P}\mathcal{P}_{U_A}(e_{i_1}e_{j_1}^T)\rangle|^2\\
&=&\|\mathcal{P}_{U_A}\mathcal{P}\mathcal{P}_{U_A}(e_{i_1}e_{j_1}^T)\|_F^2\\\
&\leq&\|\mathcal{P}_{U_A}\mathcal{P}\|^2\|\mathcal{P}_{U_A}(e_{i_1}e_{j_1}^T)\|_F^2\\
&\leq&\frac{4\mu_1(A)r_A}{m}.
\end{eqnarray*}
Then the proof is finished by using Bernstein's inequality, which states that for a collection of uniformly bounded independent random variables $\{y_i\}_{i=1}^p$ with $|y_i-\mathbb{E}(y_i)|<c$,
\begin{eqnarray*}
Pr\left(\left|\sum_{i=1}^p(y_i-\mathbb{E}(y_i))\right|>t\right)\leq{}\exp\left(-\frac{0.5t^2}{\sum_{i=1}^pVar(y_i)+ct/3}\right).
\end{eqnarray*}
Thus we have
\begin{eqnarray*}
Pr(|[Q]_{i_1j_1}|>\epsilon)&\leq&\exp\left(-\frac{0.5\epsilon^2}{4\frac{\mu_1(A)r_A}{m}+2\epsilon\frac{\mu_1(A)r_A}{3m}}\right)\\
&\leq& \exp\left(-\frac{1.5\epsilon^2m}{(12+2\epsilon)\mu_1(A)r_A}\right).
\end{eqnarray*}
By union bound,
\begin{eqnarray*}
Pr\left(\|Q\|_{\infty}\leq\epsilon\right)&\geq{}&1-n_1^2\exp\left(-\frac{1.5\epsilon^2m}{(12+2\epsilon)\mu_1(A)r_A}\right)\\
&\geq&1-n_1^{-10},
\end{eqnarray*}
provided that $r_A\leq{}\epsilon^2n_2/(c_a\mu_1(A)\log{}n_1)$ with $c_a\geq104$.
\end{proof}
\subsection{Dual Conditions}
It remains to prove Theorem~\ref{thm:noiseless} by two steps:
\begin{itemize}
\item[1.]\textbf{Dual Conditions:} Identify the sufficient conditions for $(Z=A^{+}L_0,S=S_0)$ to be the unique optimal solution to the LRR problem \eqref{eq:lrr}.
\item[2.]\textbf{Dual Certificates:} Show that the dual conditions can be satisfied, that is to say, construct the dual certificates.
\end{itemize}
The dual conditions are presented in the following lemma.
\begin{lemma}\label{app:lem:dual}
Let the SVD of $A^{+}L_0$ be $U\Sigma{}V^T$. Suppose $\mathcal{P}_{U_A}(U_0)=U_0$ and $U_A\cap\Omega=\{0\}$. Then $(A^{+}L_0,S_0)$ is the unique optimal solution to \eqref{eq:lrr} if there exists a matrix $F$ that obeys
\begin{eqnarray*}
 \textrm{(a)}&&UV^T=\lambda{}A^T(sign(S_0)+F),\\
 \textrm{(b)}&&\mathcal{P}_{\Omega}(F) = 0,\\
 \textrm{(c)}&&\|\mathcal{P}_{\Omega^{\bot}}(F)\|_{\infty}<1.
\end{eqnarray*}
\end{lemma}
\begin{proof}By standard convexity arguments~\citep{book:convex}, $(A^{+}L_0,S_0)$ is an optimal solution to \eqref{eq:lrr} if
\begin{eqnarray*}
0\in\partial\|A^{+}L_0\|_* -\lambda{}A^T\partial{}\|S_0\|_1.
\end{eqnarray*}
Note that $UV^T\in\partial\|A^{+}L_0\|_*$. Furthermore, (b) and (c) imply that $sign(S_0)+F\in\partial\|S_0\|_1$. Thus, the conditions (a), (b) and (c) are sufficient to conclude that $(A^{+}L_0,S_0)$ is an optimal (but may not be unique) solution to \eqref{eq:lrr}.

Next, we shall consider a feasible perturbation $(A^{+}L_0+\Delta_1,S_0-\Delta)$ and show that the objective strictly increases whenever $\Delta\neq0$. By $L_0+S_0=X=A(A^{+}L_0+\Delta_1)+S_0-\Delta$,
\begin{eqnarray*}
\Delta = A\Delta_1 &\textrm{and so}& \Delta\in{}\mathcal{P}_{U_A}.
\end{eqnarray*}
Let $H=-\mathcal{P}_{\Omega^{\bot}}(sign(\Delta))$. Then by Lemma~\ref{app:lem:basic:1}, $sign(S_0)+H$ is a subgradient of $\|S_0\|_1$.  By the convexity of the nuclear norm and $\ell_{1}$ norm,
\begin{align*}
&\|A^{+}L_0+\Delta_1\|_*+\lambda{}\|S_0-\Delta\|_1 \\
\geq&\|A^{+}L_0\|_*+\lambda{}\|S_0\|_1+\langle{}UV^T,\Delta_1\rangle-\lambda{}\langle{}sign(S_0)+H,\Delta\rangle\\
=&\|A^{+}L_0\|_*+\lambda{}\|S_0\|_1 + \langle{}UV^T-\lambda{}A^T{}sign(S_0),\Delta_1\rangle -\lambda{}\langle{}H,\Delta\rangle\\
=&\|A^{+}L_0\|_*+\lambda{}\|S_0\|_1 + \langle{}UV^T-\lambda{}A^T{}sign(S_0),\Delta_1\rangle +\lambda{}\|\mathcal{P}_{\Omega^{\bot}}(\Delta)\|_1\\
=&\|A^{+}L_0\|_*+\lambda{}\|S_0\|_1 + \lambda\langle{}A^TF,\Delta_1\rangle +\lambda{}\|\mathcal{P}_{\Omega^{\bot}}(\Delta)\|_1\\
=&\|A^{+}L_0\|_*+\lambda{}\|S_0\|_1+\lambda\langle{}F,A\Delta_1\rangle + \lambda{}\|\mathcal{P}_{\Omega^{\bot}}(\Delta)\|_1\\
=&\|A^{+}L_0\|_*+\lambda{}\|S_0\|_1+\lambda\langle{}F,\Delta\rangle + \lambda{}\|\mathcal{P}_{\Omega^{\bot}}(\Delta)\|_1\\
=&\|A^{+}L_0\|_*+\lambda{}\|S_0\|_1+\lambda\langle{}\mathcal{P}_{\Omega^{\bot}}(F),\mathcal{P}_{\Omega^{\bot}}(\Delta)\rangle + \lambda{}\|\mathcal{P}_{\Omega^{\bot}}(\Delta)\|_1\\
\geq&\|A^{+}L_0\|_*+\lambda{}\|S_0\|_1-\lambda\|\mathcal{P}_{\Omega^{\bot}}(F)\|_{\infty}\|\mathcal{P}_{\Omega^{\bot}}(\Delta)\|_1+\lambda{}\|\mathcal{P}_{\Omega^{\bot}}(\Delta)\|_1\\
=&\|A^{+}L_0\|_*+\lambda{}\|S_0\|_1+\lambda(1-\|\mathcal{P}_{\Omega^{\bot}}(F)\|_{\infty})\|\mathcal{P}_{\Omega^{\bot}}(\Delta)\|_1.
\end{align*}
By $\Delta\in\mathcal{P}_{U_A}$, $\|\mathcal{P}_{\Omega^{\bot}}(F)\|_{\infty}<1$ and the assumption $U_A\cap{}\Omega=\{0\}$, we have $\|\mathcal{P}_{\Omega^{\bot}}(\Delta)\|_1>0$. Thus, we have
\begin{eqnarray*}
\|A^{+}L_0\|_*+\lambda{}\|S_0\|_1+\lambda(1-\|\mathcal{P}_{\Omega^{\bot}}(F)\|_{\infty})\|\mathcal{P}_{\Omega^{\bot}}(\Delta)\|_1>\|A^{+}L_0\|_*+\lambda{}\|S_0\|_1
\end{eqnarray*}
strictly holds unless $\Delta=0$. As long as $A\Delta_1=0$, Theorem 4.1 of~\citep{tpami_2013_lrr} gives that $\|A^+L_0+\Delta_1\|_*>\|A^+L_0\|_*$ strictly holds unless $\Delta_1=0$. Hence, $(A^+L_0,S_0)$ is the unique optimal solution to the LRR problem \eqref{eq:lrr}.
\end{proof}
\subsection{Dual Certificates}
To construct a matrix $F$ which satisfies the dual conditions listed in Lemma~\ref{app:lem:dual}, we need the inverse of $P_{U_A}P_{\Omega^\bot}P_{U_A}$. The following lemma shows that $(P_{U_A}P_{\Omega^\bot}P_{U_A})^{-1}$ is well defined and has a small operator norm.
\begin{lemma}\label{app:lem:inverse}
If $\|P_{U_A}P_{\Omega}\|<1$, then the operator $\mathcal{P}_{U_A}\mathcal{P}_{\Omega^\bot}\mathcal{P}_{U_A}$ is an injection from $\mathcal{P}_{U_A}$ to $\mathcal{P}_{U_A}$, and its inverse operator is given by
\begin{eqnarray*}
\mathcal{I}+\sum_{i=1}^{\infty}(\mathcal{P}_{U_A}\mathcal{P}_{\Omega}\mathcal{P}_{U_A})^i.
\end{eqnarray*}
\end{lemma}
\begin{proof}
Let $\psi\equiv\|\mathcal{P}_{U_A}\mathcal{P}_{\Omega}\|$. By $\|\mathcal{P}_{U_A}\mathcal{P}_{\Omega}\mathcal{P}_{U_A}\|=\|\mathcal{P}_{U_A}\mathcal{P}_{\Omega}\|^2=\psi^2<1$, we have that $\mathcal{I}+\sum_{i=1}^{\infty}(\mathcal{P}_{U_A}\mathcal{P}_{\Omega}\mathcal{P}_{U_A})^i$ is well defined and has an operator norm not larger than $1/(1-\psi^2)$.

Note that
\begin{eqnarray*}
\mathcal{P}_{U_A}\mathcal{P}_{\Omega^\bot}\mathcal{P}_{U_A}&=&\mathcal{P}_{U_A}(\mathcal{I}-\mathcal{P}_{\Omega})\mathcal{P}_{U_A}\\
&=&\mathcal{P}_{U_A}(\mathcal{I}-\mathcal{P}_{U_A}\mathcal{P}_{\Omega}\mathcal{P}_{U_A}).
\end{eqnarray*}
Thus for any $M\in\mathcal{P}_{U_A}$ the following holds:
\begin{eqnarray*}
&&\mathcal{P}_{U_A}\mathcal{P}_{\Omega^\bot}\mathcal{P}_{U_A}(\mathcal{I}+\sum_{i=1}^{\infty}(\mathcal{P}_{U_A}\mathcal{P}_{\Omega}\mathcal{P}_{U_A})^i)(M)\\
&=&\mathcal{P}_{U_A}(\mathcal{I}-\mathcal{P}_{U_A}\mathcal{P}_{\Omega}\mathcal{P}_{U_A})(\mathcal{I}+\sum_{i=1}^{\infty}(\mathcal{P}_{U_A}\mathcal{P}_{\Omega}\mathcal{P}_{U_A})^i)(M)\\
&=&\mathcal{P}_{U_A}(\mathcal{I}+\sum_{i=1}^{\infty}(\mathcal{P}_{U_A}\mathcal{P}_{\Omega}\mathcal{P}_{U_A})^i-\mathcal{P}_{U_A}\mathcal{P}_{\Omega}\mathcal{P}_{U_A}-\sum_{i=2}^{\infty}(\mathcal{P}_{U_A}\mathcal{P}_{\Omega}\mathcal{P}_{U_A})^i)(M)\\
&=&\mathcal{P}_{U_A}(\mathcal{I}+\sum_{i=1}^{\infty}(\mathcal{P}_{U_A}\mathcal{P}_{\Omega}\mathcal{P}_{U_A})^i-\sum_{i=1}^{\infty}(\mathcal{P}_{U_A}\mathcal{P}_{\Omega}\mathcal{P}_{U_A})^i)(M)\\
&=&\mathcal{P}_{U_A}(M) = M.\hspace{3in}
\end{eqnarray*}
\end{proof}

The next lemma completes the construction of the dual certificates.
\begin{lemma}\label{app:lem:f}
Let
\begin{eqnarray*}
F = \mathcal{P}_{\Omega^\bot}\mathcal{P}_{U_A}\left(\mathcal{I}+\sum_{i=1}^{\infty}(\mathcal{P}_{U_A}\mathcal{P}_{\Omega}\mathcal{P}_{U_A})^i\right)
\left(\frac{1}{\lambda}(A^T)^+UV^T-\mathcal{P}_{U_A}(sign(S_0))\right),
\end{eqnarray*}
where $U$ and $V$ are the left and right singular vectors of $A^+L_0$, respectively. If the conditions stated in \eqref{eq:succ} are obeyed, then the above $F$ using $\lambda=1/\sqrt{n_1}$ satisfies (with high probability) the dual conditions (a), (b) and (c) in Lemma \ref{app:lem:dual}.
\end{lemma}
\begin{proof}
\noindent\textbf{(a):} We have
\begin{align*}
&\lambda{}A^T(sign(S_0)+F)\\
=&\lambda{}A^Tsign(S_0)+\lambda{}A^T\mathcal{P}_{U_A}(F)\\
=&\lambda{}A^Tsign(S_0)+\lambda{}A^T\mathcal{P}_{U_A}\mathcal{P}_{\Omega^\bot}\mathcal{P}_{U_A}(\mathcal{I}+\sum_{i=1}^{\infty}(\mathcal{P}_{U_A}\mathcal{P}_{\Omega}\mathcal{P}_{U_A})^i)(\frac{1}{\lambda}(A^T)^+UV^T-\mathcal{P}_{U_A}(sign(S_0)))\\
=&\lambda{}A^Tsign(S_0)+\lambda{}A^T(\frac{1}{\lambda}(A^T)^+UV^T-\mathcal{P}_{U_A}(sign(S_0)))\\
=&\lambda{}A^Tsign(S_0)+V_AV_A^TUV^T-\lambda{}A^T\mathcal{P}_{U_A}(sign(S_0))\\
=&\lambda{}A^Tsign(S_0)-\lambda{}A^T\mathcal{P}_{U_A}(sign(S_0))+V_AV_A^TUV^T\\
=&V_AV_A^TUV^T=UV^T,
\end{align*}
where the last equality follows from Theorem 4.3 of~\citep{tpami_2013_lrr}.\\

\noindent \textbf{(b):} It is easy to verify that $\mathcal{P}_{\Omega}(F)=0$.\\

\noindent\textbf{(c):} Let $\mathcal{P}=\mathcal{I}+\sum_{i=1}^{\infty}(\mathcal{P}_{U_A}\mathcal{P}_{\Omega}\mathcal{P}_{U_A})^i$ and $F={}\mathcal{P}_{\Omega^\bot}(F_1-F_2)$, where
\begin{align*}
F_1&=\mathcal{P}_{U_A}\mathcal{P}\mathcal{P}_{U_A}(\frac{1}{\lambda}(A^T)^+UV^T),\hspace{0.3in} F_2 = \mathcal{P}_{U_A}\mathcal{P}\mathcal{P}_{U_A}(sign(S_0)).
\end{align*}
In the following, we shall bound the sup-norm of each term individually.

The proof for $\|F_2\|_{\infty}$ needs to access the distribution of $sign(S_0)$. When the signs of the nonzero entries of $S_0$ are Bernoulli $\pm1$ values, i.e., $sign(S_0)=\mathcal{P}_{\Omega}(\Psi)$ with $\Psi$ being a random sign matrix as in Lemma \ref{lem:inf:pas0}, we have indeed proven
\begin{align*}
\|F_2\|_{\infty}=\|\mathcal{P}_{U_A}\mathcal{P}\mathcal{P}_{U_A}\mathcal{P}_{\Omega}(\Psi)\|_{\infty}<\epsilon,
\end{align*}
provided that $\|\mathcal{P}\|\leq1/(1-\rho_0-\epsilon)\leq2$, which follows from the condition of $\rho_0<0.5-\epsilon$.

So it remains to prove that
\begin{eqnarray*}
\|F_1\|_{\infty}<1-\epsilon.
\end{eqnarray*}
This seems easy because we could set $\lambda\rightarrow{}+\infty$. Nevertheless, to prove our main result, Theorem~\ref{thm:noiseless}, with $\lambda=1/\sqrt{n_1}$ (which is a good choice in general), one essentially needs to establish an accurate bound for $\|F_1\|_{\infty}$. Even more, the golfing scheme widely adopted by previous literatures is indeed not easy to work with in our setting. Fortunately, we can make use of the particular structure of $F_1$ and devise a simple approach to accomplish the proof. Our idea is based on the following observation: For any matrix $Q$, the $(i_1,j_1)$th entry of the matrix $\mathcal{P}_{U_A}\mathcal{P}_{\Omega}(Q)$ is
\begin{eqnarray*}
[\mathcal{P}_{U_A}\mathcal{P}_{\Omega}(Q)]_{i_1j_1}=\sum_{i,j}\delta_{ij}[Q]_{ij}\langle{}e_ie_j^T,\mathcal{P}_{U_A}(e_{i_1}e_{j_1}^T)\rangle=\sum_{i}\delta_{ij_1}[Q]_{ij_1}[U_AU_A^T]_{ii_1},
\end{eqnarray*}
which reveals the fact that the absolute value of $[F_1]_{i_1j_1}$ closely relates to the length of the $j_1$th column of $(A^T)^+UV^T$. So it may not lose much accuracy to use the relaxation of $\|\cdot\|_{\infty}\leq\|\cdot\|_{2,\infty}$. For the sake of consistency, we use the $\ell_{2,\infty}$ norm to define as follows the third coherence parameter of $L_0$, associating with a dictionary matrix $A$:
\begin{defn}
For $L_0\in\mathbb{R}^{m\times{}n}$ of rank $r_0$, its third coherence parameter, associating with a non-orthonormal, column-wisely unit-normed dictionary matrix $A$ which also satisfies $\mathcal{P}_{U_A}(U_0)=U_0$, is defined as
\begin{eqnarray}\label{eq:u3a}
\mu_3^A(L_0)= \frac{n^2(\|(A^T)^+UV^T\|_{2,\infty})^2}{(\log{}n)^2r_0\gamma_A},
\end{eqnarray}
where $U$ and $V$ are the left and right singular vectors of $A^+L_0$, respectively, and $\gamma_A$ is the condition number of the matrix $A$.
\end{defn}
\begin{figure}
\begin{center}
\includegraphics[width=0.98\textwidth]{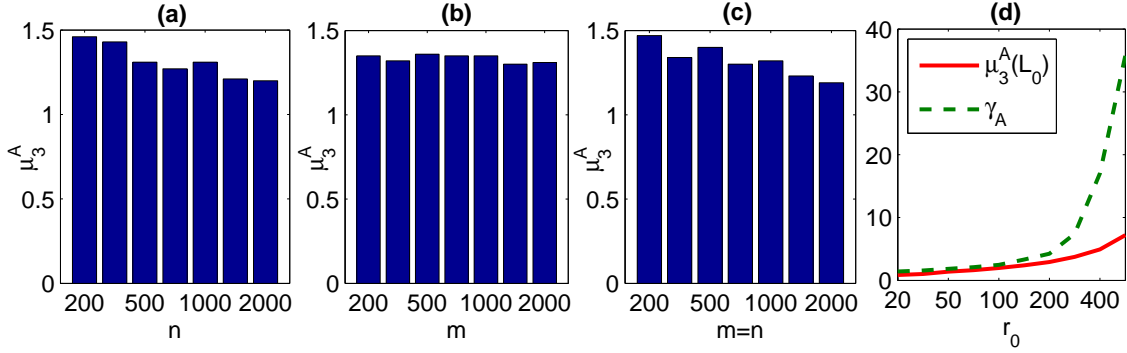}\vspace{-0.2in}
\end{center}
\caption{Investigating the properties of $\mu_3^A(L_0)$. (a) $m=500$ and $r_0=50$ are fixed, while $n$ is varying. (b) $n=500$ and $r_0=50$ are fixed, while $m$ is varying. (c) $r_0=50$ is fixed, while $m$ and $n$ are varying ($m=n$). (d) $m=n=500$ are fixed, while $r_0$ is varying. In these experiments, the dictionary $A$ is $\mathcal{P}_{U_0}(R)$ with normalized columns, where $R$ is an $m\times{}n$ random Gaussian matrix. The numbers shown in above figures are averaged from 10 random trials.}\label{fig:u3a}
\end{figure}

Figure~\ref{fig:u3a} demonstrates some properties about this particular coherence parameter, $\mu_3^A$. It can be seen that $\mu_3^A$ is approximately a numerical constant equaling to 1, as long as the rank is not too high such that the dictionary matrix $A$ is well-conditioned.

By Lemma~\ref{lem:papo}, $\|\mathcal{P}\|\leq1/(1-\rho_0-\epsilon)$. By \eqref{eq:u3a},
\begin{eqnarray*}
\|(A^T)^+UV^T)\|_{2,\infty}\leq\frac{\sqrt{\gamma_A\mu_3^A(L_0)r_0}\log{}n}{n}.
\end{eqnarray*}
Thus we have
\begin{eqnarray*}
\|F_1\|_{\infty}&=&\|U_AU_A^T\mathcal{P}\mathcal{P}_{U_A}(\frac{1}{\lambda}(A^T)^+UV^T)\|_{\infty}\\
&\leq&\max_{i}\|e_i^TU_A\|_2\|U_A^T\mathcal{P}\mathcal{P}_{U_A}(\frac{1}{\lambda}(A^T)^+UV^T)\|_{2,\infty}\\
&\leq&\sqrt{\frac{\mu_1(A)r_A}{m}}\|\mathcal{P}\|\|(\frac{1}{\lambda}(A^T)^+UV^T)\|_{2,\infty}\\
&\leq&\frac{\sqrt{\mu_1(A)\mu_3^A(L_0)\gamma_Ar_Ar_0}\log{}n}{\lambda\sqrt{m}n(1-\rho_0-\epsilon)}.
\end{eqnarray*}
By $r_0\leq{}r_A\leq\epsilon^2n_2/(c_a\mu_1(A)\log{}n_1)$ and setting $\lambda=\sqrt{\mu_3^A(L_0)\gamma_A/(\mu_1(A)n_1)}$,
\begin{eqnarray*}
\|F_1\|_{\infty}\leq\frac{\epsilon^2n_2\sqrt{n_1}\log{n}}{c_a(1-\rho_0-\epsilon)\sqrt{m}n\log{n_1}}\leq\frac{\epsilon^2}{c_a(1-\rho_0-\epsilon)}.
\end{eqnarray*}
Since $\mu_3^A(L_0)\gamma_A/\mu_1(A)\approx1$ (provided that $A$ is well-conditioned), we claim $\lambda=1/\sqrt{n_1}$ for the sake of simplicity\footnote{This detail also suggests that $\lambda=1/\sqrt{n_1}$ may not be the ``best'' choice.}.

Now the dual condition $\|\mathcal{P}_{\Omega^{\bot}}(F)\|_{\infty}<1$ is proved by
\begin{eqnarray*}
&&\|F\|_{\infty}<1,\\
&\leftarrow&\frac{\epsilon^2}{c_a(1-\rho_0-\epsilon)}<1-\epsilon,\\
&\leftarrow&\rho_0<1-2\epsilon,\\
&\leftarrow&\rho_0<0.5-\epsilon.
\end{eqnarray*}
We claim $\rho_0<0.5-\epsilon$ instead of $\rho_0<1-2\epsilon$ because Lemma~\ref{lem:inf:pas0} requires $\|\mathcal{P}\|\leq2$, which follows from $\rho_0<0.5-\epsilon$.
\end{proof}
\section{Experiments}\label{sec:exp}
Our main result, Theorem~\ref{thm:noiseless}, is useful in both supervised and unsupervised environments. For the fair of comparison, in the experiments of this paper we shall focus on demonstrating the superiorities of our unsupervised Algorithm~\ref{alg:mr} over RPCA.

\subsection{Results on Randomly Generated Matrices}
We first verify the effectiveness of our Algorithm \ref{alg:mr} on randomly generated matrices. We generate a collection of $200\times1000$ data matrices according to the model of $X=\mathcal{P}_{\Omega^\bot}(L_0)+\mathcal{P}_{\Omega}(S_0)$: $\Omega$ is a support set chosen at random; $L_0$ is created by sampling 200 data points from each of 5 randomly generated subspaces, and its values are normalized such that $\|L_0\|_{\infty}=1$; $S_0$ is consisting of random values from Bernoulli $\pm1$. The dimension of each subspace varies from 1 to 20 with step size 1, and thus the rank of $L_0$ varies from 5 to 100 with step size 5. The fraction $|\Omega|/(mn)$ varies from 2.5\% to 50\% with step size 2.5\%. For each pair of rank and support size $(r_0,|\Omega|)$, we run 10 trials, resulting in a total of 4000 ($20\times20\times10$) trials.
\begin{figure}[h!]
\begin{center}
\includegraphics[width=0.95\textwidth]{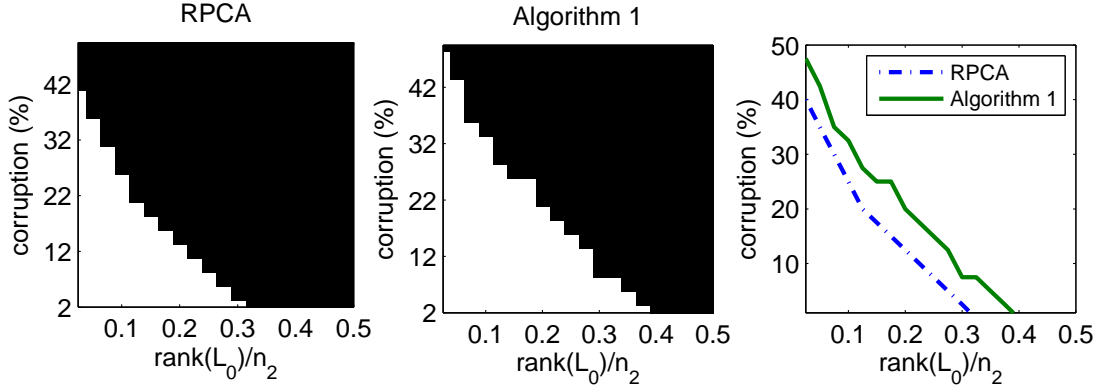}\vspace{-0.3in}
\end{center}
\caption{Algorithm \ref{alg:mr} vs RPCA on recovering randomly generated matrices, both using $\lambda=1/\sqrt{n_1}$. A curve shown in the third subfigure is the boundary for a method to be successful --- The recovery is successful for any pair $(r_0/n_2,|\Omega|/(mn))$ that locates below the curve. Here, the success is in a sense that $\|\hat{L}_0-L_0\|_F<0.05\|L_0\|_F$, where $\hat{L}_0$ denotes an estimate of $L_0$.}\label{fig:recover}
\end{figure}

Figure~\ref{fig:recover} compares our Algorithm~\ref{alg:mr} to RPCA, both using $\lambda=1/\sqrt{n_1}$. It can be seen that the learnt dictionary matrix works distinctly better than the identity dictionary adopted by RPCA. Namely, the success area (i.e., the area of the white region) of our algorithm is 46\% wider than that of RPCA! One may have noticed that RPCA owns a region to be exactly successful. This is because in these experiments the coherence parameters are not too large, namely $\mu_1(L_0)\leq3.5$ and $\mu_2(L_0)\leq13.7$. Whenever $\mu_2$ reaches the upper bound $n$, e.g., the example shown in Figure~\ref{fig:demo}, the success region of RPCA will vanish.
\subsection{Results on Corrupted Motion Sequences}
We now experiment with 11 additional sequences attached to the Hopkins155~\citep{hopkin155} database. In those sequences, about 10\% of the entries in the data matrix of trajectories are unobserved (i.e., missed) due to visual occlusion. We replace each missing entry with a number from Bernoulli $\pm1$, resulting in a collection of corrupted trajectory matrices for evaluating the effectiveness of matrix recovery algorithms. We perform subspace clustering on both the corrupted trajectory matrices and the recovered versions, and use the clustering error rates produced by existing subspace clustering methods as the evaluation metrics. We consider three state-of-the-art subspace clustering methods: Shape Interaction Matrix (SIM)~\citep{ijcv_1998_factor}, Low-Rank Representation with $A=X$~\citep{icml_2010_lrr} (which is referred to as ``LRRx'') and Sparse Subspace Clustering (SSC)~\citep{cvpr_2009_ssc}.
\begin{table}[h!]
\caption{Clustering error rates (\%) on 11 corrupted motion sequences.}\label{tb:motion}
\begin{center}
\begin{tabular}{|c|c|c|c|c|c|c|}\hline
                            & Mean  & Median & Maximum & Minimum & Std. & Time (sec.)\\\hline
SIM                         &29.19  & 27.77  & 45.82   & 12.45   &11.74     &0.07\\
RPCA + SIM                  &14.82  & 8.38   & 45.78   & 0.97    &16.23     &9.96\\
Algorithm~\ref{alg:mr} + SIM  &8.74   & 3.09   & 42.61   & 0.23    &12.95     &11.64\\\hline
LRRx                        &21.38  & 22.00  & 56.96   & 0.58    &17.10     &1.80\\
RPCA + LRRx                 &10.70  & 3.05   & 46.25   & 0.20    &15.63     &10.75\\
Algorithm~\ref{alg:mr} + LRRx &7.09   & 3.06   & 32.33   & 0.22    &10.59     &12.11\\\hline
SSC                         &22.81  & 20.78  & 58.24   & 1.55    &18.46     &3.18\\
RPCA + SSC                  &9.50   & 2.13   & 50.32   & 0.61    &16.17     &12.51\\
Algorithm~\ref{alg:mr} + SSC  &5.74   & 1.85   & 27.84   & 0.20    &8.52      &13.11\\\hline
\end{tabular}
\end{center}
\end{table}

Table~\ref{tb:motion} shows the error rates of various algorithms. Without the preprocessing of matrix recovery, all the subspace clustering methods fail to accurately categorize the trajectories of motion objects, producing error rates higher than 20\%. This illustrates that it is important for motion segmentation to correct the gross corruptions possibly existing in the data matrix of trajectories. By using RPCA ($\lambda=1/\sqrt{n_1}$) to correct the corruptions, the clustering performances of all considered methods are improved dramatically. For example, the error rate of SSC is reduced from 22.9\% to 9.5\%. By choosing a better dictionary (than the identity) for LRR ($\lambda=1/\sqrt{n_1}$), the error rates can be reduced again, namely from $9.5\%$ to $5.7\%$, which is a 40\% improvement. These results verify the effectiveness of our dictionary learning strategy in realistic environments.
\section{Conclusion and Future Work}\label{sec:con}
In this paper, we studied the problem of disentangling the low-rank ($L_0$) and sparse ($S_0$) components in a given data matrix. Whenever the low-rank component owns some extra structures, the state-of-the-art RPCA method might fail even if $L_0$ is strictly low-rank. As a typical example, consider the case where there is a mixture structure of multiple subspaces underlying $L_0$. When the subspace (i.e., cluster) number goes large, the second coherence parameter will enlarge and thus the performance of RPCA degrades. To overcome the challenges arising from coherent data, theoretically, one needs to capture the extra structures that produce high coherence. Nevertheless, such a strategy suffers several practical issues and is therefore infeasible. In sharp contrast, it is much simpler to solve the problem by LRR: When the dictionary matrix $A$ in LRR satisfies certain conditions, namely $A$ is low-rank and $U_0\subset{}U_A$, LRR can avoid the second coherence parameter that has potential to be large. Furthermore, we established a heuristic algorithm that utilizes RPCA to approximately pursue a qualified dictionary. Experimental results showed that our algorithm performed better than RPCA. However, there still remain several problems for future work.
\begin{itemize}\vspace{-0.08in}
\item[$\diamond$] By $AZ^*=L_0$, the column space of the dictionary $A$ approximately has the same properties as $L_0$, and thus, roughly, $\mu_1(A)\approx\mu_1(L_0)$. So this paper still needs to assume that the first coherence parameter $\mu_1$ is small and only addresses the cases where the second coherence parameter $\mu_2$ might be large. In some domains such as the text documents, both the row space and column space can own some clustering structures, and thus both $\mu_1$ and $\mu_2$ can be large. New models are required to well handle such coherent data.\vspace{-0.08in}
\item[$\diamond$] It is possible to prove that Algorithm~\ref{alg:mr} is superior over RPCA in theory, because the conditions (i.e., $A$ is low-rank and $U_0\subset{}U_A$) required by Algorithm~\ref{alg:mr} to succeed are weaker than $A=L_0$. It is significant to accurately identify in which conditions RPCA can produce a solution that is able to meet those conditions.\vspace{-0.08in}
\item[$\diamond$] While theorem~\ref{thm:noiseless} points out a generic direction for learning the dictionary matrix in LRR, the specific learning procedure is not unique and our Algorithm~\ref{alg:mr} is not exclusive either. For example, one may drive some kind of optimization framework to jointly compute the variables $A$ and $Z$.
\end{itemize}

\vspace{-0.2in}

\section*{Acknowledgement}
Guangcan Liu is a Postdoctoral Researcher supported by NSF-DMS0808864, NSF-SES1131848,  NSF-EAGER1249316, AFOSR-FA9550-13-1-0137, and ONR-N00014-13-1-0764. Ping Li is also partially supported by  NSF-III1360971 and NSF-BIGDATA1419210.

\vspace{-0.1in}
\appendix
\section{List of Notations}\label{sec:app:notations}
\begin{center}
\begin{tabular}{ll}
$(\cdot)^{+}$ & Moore-Penrose pseudoinverse of a matrix. \\
$\otimes$ & Kronecker product.\\
$e_i$ & The $i$th standard basis. \\
$[\cdot]_{ij}$ & The $(i,j)$th entry of a matrix. \\
$X\in\mathbb{R}^{m\times{}n}$             &   The observed data matrix.\\
$A$, $U_A\Sigma_AV_A^T$                   &   The dictionary matrix, and its SVD\\
$L_0$, $U_0\Sigma_0V_0^T$                 &   The ground truth of the data matrix, and its SVD.\\
$S_0\in\mathbb{R}^{m\times{}n}$           &   The ground truth of the corruption matrix.\\
$U\Sigma{}V^T$ & The SVD of $A^{+}L_0$.\\
$r_0,r_A$ & The ranks of $L_0$ and $A$.\\
$\gamma_A$ & The condition number of $A$.\\
$\mu_1,\mu_2,\mu_3$ & The first, second and third coherence parameters of a matrix.\\
$\mu_3^A(\cdot)$ & The third coherence parameter of a matrix, associating with  $A$.\\
$n_1,n_2$ & $n_1=\max(m,n)$,$n_2 = \min(m,n)$.\\
$\Omega$  &   Locations of the nonzero entries of $S_0$.\\
$\Omega^c$ & The complement of $\Omega$.\\
$\mathcal{P}_{U_0}$, $\mathcal{P}_{V_0}$ & The projections onto the space spanned by $U_0$ (resp. $V_0$).\\
$\mathcal{P}_{\Omega}$, $\mathcal{P}_{\Omega^\bot}$ & The projections onto the space of matrices supported on $\Omega$ (resp. $\Omega^c$).\\
$\Id,\mathcal{I}$  & The identity matrix and the identity operator.\\
$|\Omega|$ & The cardinality of $\Omega$, i.e., the number of nonzero entries in $S_0$. \\
$sign(\cdot)$ &The signum function. \\
$\partial$ & The subgradient of a function. \\
$\|\cdot\|_2$ & The $\ell_2$ norm of a vector.\\
$\langle{}\cdot\rangle$ & The inner product of two matrices or vectors. \\
$\|\cdot\|$ & The operator norm or 2-norm of a matrix, i.e., largest singular value.\\
$\|\cdot\|_*$ & The nuclear norm of a matrix.\\
$\|\cdot\|_F$ & The Frobenius norm of a matrix.\\
$\|\cdot\|_{2,\infty}$ & The $\ell_{2,\infty}$ norm, i.e., the largest $\ell_2$ norm of the columns of a matrix.\\
$\|\cdot\|_1$ & The $\ell_1$ norm of a matrix seen as a long vector. \\
$\|\cdot\|_{\infty}$ & The sup-norm of a matrix seen as a long vector.\\
$Ber(\rho)$ & A Bernoulli distribution with expected value $\rho$ and variance $\rho(1-\rho)$.
\end{tabular}
\end{center}
\section{Why Does $\mu_2$ Increase with the Cluster Number?}\label{app:why}
\subsection{Zipf's Law}
When the data points are sampled from a low-rank subspace \emph{uniformly at random}, it has been proven by~\citep{Candes:2009:math} that the first and second coherence parameters are bounded. Namely, $\mu_1(L_0)\leq{}c$ and $\mu_2(L_0)\leq{}c$ for some numerical constant $c$ independent of the characteristics of $L_0$. Although correct, such a property is not enough to interpret the phenomenon that the coherence parameters increase with the cluster number underlying $L_0$. Hence, it is necessary to establish a more accurate rule to characterize the coherence parameters. Through extensive experiments, we find that the first and second coherence parameters actually follow the well-known Zipf's law. More precisely, if the data points (which form the column vectors of $L_0\in\mathbb{R}^{m\times{}n}$) are uniformly sampled from a $r_0$-dimensional subspace, then, roughly, the logarithm of coherence is inversely proportional to the logarithm of $1+r_0$. That is,
\begin{eqnarray}\label{eq:zipf}
\log(\mu_1(L_0))\log(1+r_0)\approx{}c_1 &\textrm{and}&  \log(\mu_2(L_0))\log(1+r_0)\approx{}c_2,
\end{eqnarray}
where $c_1$ and $c_2$ are two constants. The results in Figure~\ref{fig:zipf} verify the above Zipf's law. Note that the Zipf's law \eqref{eq:zipf} can also induce the boundedness property proved by~\citep{Candes:2009:math}. Namely, \eqref{eq:zipf} approximately gives that $\mu_1(L_0)\leq{}\exp(c_1/\log2)$ and $\mu_2(L_0)\leq{}\exp(c_2/\log2)$.

The above Zipf's law suggests that the coherence must be inversely proportional to the rank of data. This is intuitively interpretable. Let $y_j=[U_0]_{ij}$ and $C_{r_0}=\|U_0^Te_i\|_2^2=\sum_{j=1}^{r_0}y_j^2$. Then it can be seen that $C_{r_0}$ is the squared Euclidean length of the first $r_0$ components of a unit vector distributed on the $m$-dimensional unit sphere. With these notations, it can be seen that $\mu_1$ is the largest order statistic of $C_{r_0}$ divided by the expectation of $C_{r_0}$:
\begin{eqnarray*}
\mu_1(L_0)=\frac{m}{r_0}\max_i\|U_0^Te_i\|_2^2 = \frac{\max_i\|U_0^Te_i\|_2^2}{\frac{r_0}{m}}=\frac{\max(C_{r_0})}{\mathbb{E}(C_{r_0})}.
\end{eqnarray*}
Now it is unfolded that the first (and second) coherence parameter of a matrix with rank $r_0$ is actually some kind of uncertainty of the first $r_0$ components of a unit-normed, $m$-dimensional random vector. Thus if $r_0=m$ (i.e., $L_0$ is full rank), then the uncertainty vanishes and $\mu_1(L_0)=1$. Similarly if $r_0=1$, the uncertainty measured by $\max(C_{r_0})/\mathbb{E}(C_{r_0})$ is as high as that of a single random number.

The Zipf's law \eqref{eq:zipf} is useful, because it provides us a trackable approach to estimate the coherence parameters when the data points are \emph{not} uniformly sampled, as will be shown in the next section.
\begin{figure}
\begin{center}
\includegraphics[width=0.88\textwidth]{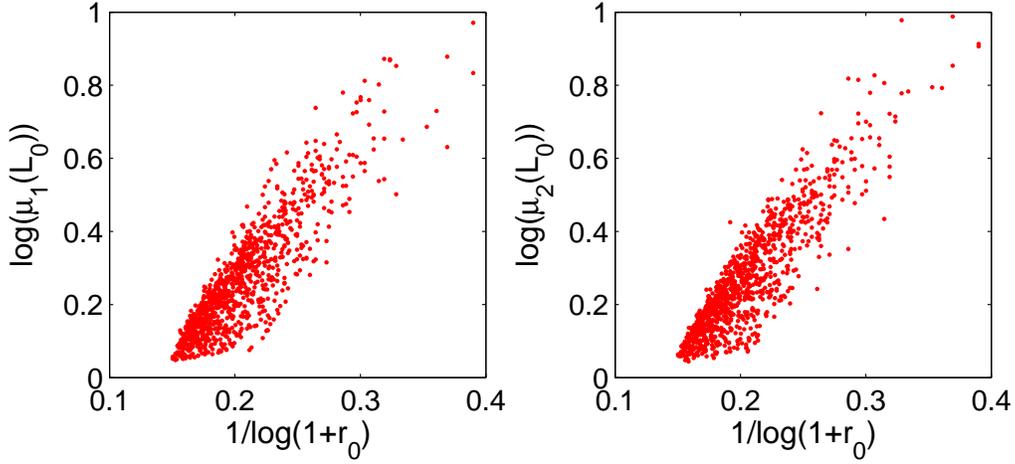}
\caption{Verifying the Zipf's law, using one million randomly generated matrices. The height $m$ and width $n$ of $L_0$ are random integers from the range 100 to 1000. The rank of $L_0$ is set as $r_0=h\min\{m,n\}$, where $h$ is a random number from the interval (0.1,0.9). For the clarity of viewing, we randomly select 10,000 out of one million simulation results to show. For all one million simulations, we have calculated that $\mathbb{E}(c_1)=\mathbb{E}(c_2)=1.23$ and $Std(c_1)=Std(c_2)=0.41$.}\label{fig:zipf}
\end{center}
\end{figure}
\subsection{An Explanation to the $\mu_2$-phenomenon}
Ideally, if the values in $U_0$ and $V_0$ are perfectly spreading out, namely $[U_0]_{ij}=[U_0]_{i_1j_1}$ and $[V_0]_{ij}=[V_0]_{i_1j_1}, \forall{}i,j,i_1,j_1$, then $\mu_1(L_0)=\mu_2(L_0)=1$. However, this is unlikely for $\mu_2(L_0)$ to happen, as it is provable that the row projector $V_0V_0^T$, which is also known as Shape Interaction Matrix (SIM) in subspace clustering, measures the subspace membership of the data points~\citep{ijcv_1998_factor,tpami_2013_lrr}. More precisely, if the data points in $L_0$ are sampled from $k$ number of \emph{independent} subspaces, saying $L_0=[L_0^{(1)},\cdots,L_0^{(k)}]$, where $L_0^{(i)}$ with SVD $U_i\Sigma_iV_i^T$ is a matrix of data points from the $i$th subspace, then $V_0$ is equivalent to a block-diagonal matrix that has nonzero entries only on $k$ number of blocks:
\begin{eqnarray*}
V_0\sim\left[\begin{array}{cccc}
V_1&0&0&0\\
0&V_2&0&0\\
0&0&\ddots&0\\
0&0&0&V_k\\
\end{array}\right].
\end{eqnarray*}
In this case, it is demonstrable that the second coherence parameter $\mu_2(L_0)$ depends on the cluster number $k$. For the convenience of analysis, we assume that the dimensions of all subspaces are equal, i.e., $\rank{L_0^{(i)}}=r_0/k,\forall{},i=1,\cdots,k$, and the sampling in each subspace is \emph{uniform}. Then the Zipf's law \eqref{eq:zipf} gives
\begin{eqnarray}\label{eq:u2k}
\mu_2(L_0) =\max_{i}\mu_2(L_0^{(i)})\approx\exp(\frac{c_2}{\log(1+\frac{r_0}{k})}),
\end{eqnarray}
where $k$ is the cluster number. Hence, approximately, the second coherence parameter $\mu_2(L_0)$ will increase with the cluster number underlying $L_0$.
\section{Proof of Theorem 2}\label{app:proof:noisy}
\begin{proof} Let $(Z^*,S^*)$ be an optimal solution to \eqref{eq:lrr:noisy}. Denote $N_L=AZ^*-L_0$, $N_S=S^*-S_0$ and $E=N_L+N_S$. Then we have
\begin{eqnarray*}
\|E\|_F &=& \| (X - L_0 - S_0) - (X-AZ^*-S^*)\|_F \\
&\leq& \| (X - L_0 - S_0)\|_F + \|(X-AZ^*-S^*)\|_F\\
&\leq& 2\varepsilon.
\end{eqnarray*}
Provided that $|\Omega|<(0.35-\epsilon)mn$, the proof process of Lemma~\ref{app:lem:f} shows that
\begin{eqnarray*}
\|F\|_{\infty}<0.5.
\end{eqnarray*}
By the optimality of $(Z^*,S^*)$,
\begin{eqnarray*}
\|A^+L_0\|_*+\lambda\|S_0\|_1&\geq&\|Z^*\|_*+\lambda\|S^*\|_1\\
&\geq&\|A^+L_0\|_*+\lambda\|S_0\|_1 + \langle{}UV^T,Z^*-A^+L_0\rangle+\lambda\langle{}sign(S_0)+H,N_S\rangle\\
&=&\|A^+L_0\|_*+\lambda\|S_0\|_1 + \lambda\langle{}sign(S_0)+F,N_L\rangle+\lambda\langle{}sign(S_0)+H,N_S\rangle,
\end{eqnarray*}
which leads to
\begin{eqnarray*}
0&\geq&\langle{}sign(S_0)+F,N_L\rangle+\langle{}sign(S_0)+H,N_S\rangle \\
&=& \langle{}sign(S_0)+F,N_L\rangle+\langle{}sign(S_0)+H,E-N_L\rangle\\
&=& \langle{}F-H,N_L\rangle+\langle{}sign(S_0)+H,E\rangle\\
&\geq&0.5\|\mathcal{P}_{\Omega^\bot}(N_L)\|_1 - \|E\|_1.
\end{eqnarray*}
Hence,
\begin{eqnarray*}
\|\mathcal{P}_{\Omega^\bot}(N_L)\|_F&\leq&\|\mathcal{P}_{\Omega^\bot}(N_L)\|_1\leq2\|E\|_1\\
&\leq&2\sqrt{mn}\|E\|_F\leq4\sqrt{mn}\varepsilon.
\end{eqnarray*}
By $N_L=AZ^*-L_0\in\mathcal{P}_{U_A}$,
\begin{eqnarray*}
N_L =  \mathcal{P}\mathcal{P}_{U_A}\mathcal{P}_{\Omega^\bot}(N_L),
\end{eqnarray*}
where $\mathcal{P}=\mathcal{I}+\sum_{i=1}^{\infty}(\mathcal{P}_{U_A}\mathcal{P}_{\Omega}\mathcal{P}_{U_A})^i$. By $\|\mathcal{P}\|\leq2$,
\begin{eqnarray*}
\|N_L\|_F&\leq&\|\mathcal{P}\|\|\mathcal{P}_{U_A}\mathcal{P}_{\Omega^\bot}(N_L)\|_F\leq\|\mathcal{P}\|\|\mathcal{P}_{\Omega^\bot}(N_L)\|_F\\
&\leq&8\sqrt{mn}\varepsilon.
\end{eqnarray*}
\end{proof}
\section{Optimization Procedure}
\begin{algorithm}[tb]
   \caption{Solving Problem \eqref{eq:lrr} by Exact ALM}\label{alg:alm:lrr}
\begin{algorithmic}
   \STATE {\bfseries Input:} data matrix $X$, dictionary matrix $A$, parameter $\lambda$.
   \STATE {\bfseries Initialization:} $Z=J=0,S=0,Y=0,W=0,\theta=0.1,\tau=5$.
   \WHILE{not converged}
   \STATE \textbf{1.} Alternating minimization:
   \WHILE{not converged}
   \STATE \textbf{1.1.} fix the others and update $J$ by $$J=\arg\min\frac{1}{\theta}||J||_*+\frac{1}{2}||J-(Z+W/\theta)||_F^2.$$
   \STATE \textbf{1.2.} fix the others and update $Z$ by
   $$Z=(\Id+A^TA)^{-1}(A^T(X-S)+J+(A^TY-W)/\theta).$$
   \STATE \textbf{1.3.} fix the others and update $S$ by $$S = \arg\min\frac{\lambda}{\theta}\|S\|_{1}+\frac{1}{2}||S-(X-AZ+Y/\theta)||_F^2.$$
  \ENDWHILE
   \STATE \textbf{2.} update the Lagrange multipliers and the parameter $\theta$
   \begin{eqnarray*}
   Y &=& Y + \theta(X-AZ-S),\\
  W &=& W + \theta(Z-J),\\
   \theta&=&\theta\tau.
 \end{eqnarray*}
   \ENDWHILE
\end{algorithmic}
\end{algorithm}
In this work, we use the exact ALM method to solve the optimization problem \eqref{eq:lrr}. We first convert \eqref{eq:lrr} to the following equivalent problem:
\begin{eqnarray*}
\min_{Z,S,J} \norm{J}_*+\lambda{\norm{S}_{1}},\textrm{ s.t. } X = AZ+S, Z=J.
\end{eqnarray*}
This problem can be solved by the ALM method, which minimizes the following augmented Lagrange function:
\begin{eqnarray*}
\norm{J}_*+\lambda{\norm{S}_{1}} + \langle{}Y,X-AZ-S\rangle+\langle{}W,Z-J\rangle+\frac{\theta}{2}(\norm{X-AZ-S}_F^2+\norm{Z-J}_F^2)
\end{eqnarray*}
with respect to $J$, $Z$ and $S$, respectively, by fixing the other variables, and then updating the Lagrange multipliers $Y$ and $W$. Algorithm \ref{alg:alm:lrr} summarizes the whole procedure of the optimization procedure.


\newpage\clearpage

\begin{thebibliography}{28}
\providecommand{\natexlab}[1]{#1}
\providecommand{\url}[1]{\texttt{#1}}
\expandafter\ifx\csname urlstyle\endcsname\relax
  \providecommand{\doi}[1]{doi: #1}\else
  \providecommand{\doi}{doi: \begingroup \urlstyle{rm}\Url}\fi

\bibitem[Borcea et~al.(2012)Borcea, Callaghan, and Papanicolaou]{sar:2012}
Liliana Borcea, Thomas Callaghan, and George Papanicolaou.
\newblock Synthetic aperture radar imaging and motion estimation via robust
  principle component analysis.
\newblock \emph{Arxiv}, 2012.

\bibitem[Cand{\`e}s and Plan(2010)]{CandesPIEEE}
Emmanuel Cand{\`e}s and Yaniv Plan.
\newblock Matrix completion with noise.
\newblock In \emph{IEEE Proceeding}, volume~98, pages 925--936, 2010.

\bibitem[Cand{\`e}s and Recht(2009)]{Candes:2009:math}
Emmanuel Cand{\`e}s and Benjamin Recht.
\newblock Exact matrix completion via convex optimization.
\newblock \emph{Foundations of Computational Mathematics}, 9\penalty0
  (6):\penalty0 717--772, 2009.

\bibitem[Cand\`{e}s et~al.(2011)Cand\`{e}s, Li, Ma, and
  Wright]{Candes:2009:JournalACM}
Emmanuel~J. Cand\`{e}s, Xiaodong Li, Yi~Ma, and John Wright.
\newblock Robust principal component analysis?
\newblock \emph{Journal of the ACM}, 58\penalty0 (3):\penalty0 1--37, 2011.

\bibitem[Costeira and Kanade(1998)]{ijcv_1998_factor}
Joao Costeira and Takeo Kanade.
\newblock A multibody factorization method for independently moving objects.
\newblock \emph{International Journal of Computer Vision}, 29\penalty0
  (3):\penalty0 159--179, 1998.

\bibitem[Elhamifar and Vidal(2009)]{cvpr_2009_ssc}
E.~Elhamifar and R.~Vidal.
\newblock Sparse subspace clustering.
\newblock In \emph{IEEE Conference on Computer Vision and Pattern Recognition},
  volume~2, pages 2790--2797, 2009.

\bibitem[Fazel(2002)]{phd_2002_nuclear}
M.~Fazel.
\newblock Matrix rank minimization with applications.
\newblock \emph{PhD thesis}, 2002.

\bibitem[Fischler and Bolles(1981)]{cacm_1981_ransac}
Martin Fischler and Robert Bolles.
\newblock Random sample consensus: A paradigm for model fitting with
  applications to image analysis and automated cartography.
\newblock \emph{Communications of the ACM}, 24\penalty0 (6):\penalty0 381--395,
  1981.

\bibitem[Gnanadesikan and Kettenring(1972)]{robust:1972}
R.~Gnanadesikan and J.~R. Kettenring.
\newblock Robust estimates, residuals, and outlier detection with multiresponse
  data.
\newblock \emph{Biometrics}, 28\penalty0 (1):\penalty0 81--124, 1972.

\bibitem[Gross(2011)]{gross:2011:tit}
D.~Gross.
\newblock Recovering low-rank matrices from few coefficients in any basis.
\newblock \emph{IEEE Transactions on Information Theory}, 57\penalty0
  (3):\penalty0 1548--1566, 2011.

\bibitem[Ke and Kanade(2005)]{ke:cvpr:2005}
Qifa Ke and Takeo Kanade.
\newblock Robust l$_{\mbox{1}}$ norm factorization in the presence of outliers
  and missing data by alternative convex programming.
\newblock In \emph{IEEE Conference on Computer Vision and Pattern Recognition},
  pages 739--746, 2005.

\bibitem[la~Torre and Black(2003)]{torre:ijcv:2003}
Fernando~De la~Torre and Michael~J. Black.
\newblock A framework for robust subspace learning.
\newblock \emph{International Journal of Computer Vision}, 54\penalty0
  (1-3):\penalty0 117--142, 2003.

\bibitem[Liu et~al.(2010{\natexlab{a}})Liu, Lin, Tang, and Yu]{Liu:2010:HGM}
Guangcan Liu, Zhouchen Lin, Xiaoou Tang, and Yong Yu.
\newblock Unsupervised object segmentation with a hybrid graph model (hgm).
\newblock \emph{IEEE Transactions on Pattern Analysis and Machine
  Intelligence}, 32\penalty0 (5):\penalty0 910--924, 2010{\natexlab{a}}.
\newblock ISSN 0162-8828.

\bibitem[Liu et~al.(2010{\natexlab{b}})Liu, Lin, and Yu]{icml_2010_lrr}
Guangcan Liu, Zhouchen Lin, and Yong Yu.
\newblock Robust subspace segmentation by low-rank representation.
\newblock In \emph{International Conference on Machine Learning}, pages
  663--670, 2010{\natexlab{b}}.

\bibitem[Liu et~al.(2012)Liu, Xu, and Yan]{jmlr_2012_lrr}
Guangcan Liu, Huan Xu, and Shuicheng Yan.
\newblock Exact subspace segmentation and outlier detection by low-rank
  representation.
\newblock \emph{Journal of Machine Learning Research - Proceedings Track},
  22:\penalty0 703--711, 2012.

\bibitem[Liu et~al.(2013)Liu, Lin, Yan, Sun, Yu, and Ma]{tpami_2013_lrr}
Guangcan Liu, Zhouchen Lin, Shuicheng Yan, Ju~Sun, Yong Yu, and Yi~Ma.
\newblock Robust recovery of subspace structures by low-rank representation.
\newblock \emph{IEEE Transactions on Pattern Analysis and Machine
  Intelligence}, 35\penalty0 (1):\penalty0 171--184, 2013.

\bibitem[Mazumder et~al.(2010)Mazumder, Hastie, and
  Tibshirani]{rahul:jlmr:2010}
Rahul Mazumder, Trevor Hastie, and Robert Tibshirani.
\newblock Spectral regularization algorithms for learning large incomplete
  matrices.
\newblock \emph{Journal of Machine Learning Research}, 11:\penalty0 2287--2322,
  2010.

\bibitem[Otazo et~al.(2012)Otazo, Cand{\`e}s, and Sodickson]{mri:2012}
Ricardo Otazo, Emmanuel Cand{\`e}s, and Daniel~K. Sodickson.
\newblock Low-rank and sparse matrix decomposition for accelerated dynamic mri
  with separation of background and dynamic components.
\newblock \emph{Arxiv}, 2012.

\bibitem[Peng et~al.(2012)Peng, Ganesh, Wright, Xu, and Ma]{peng:pami:2012}
YiGang Peng, Arvind Ganesh, John Wright, Wenli Xu, and Yi~Ma.
\newblock Rasl: Robust alignment by sparse and low-rank decomposition for
  linearly correlated images.
\newblock \emph{IEEE Transactions on Pattern Analysis and Machine
  Intelligence}, 34\penalty0 (11):\penalty0 2233--2246, 2012.

\bibitem[Rockafellar(1970)]{book:convex}
R.~Rockafellar.
\newblock \emph{Convex {A}nalysis}.
\newblock Princeton University Press, Princeton, NJ, USA, 1970.

\bibitem[Rudelson(1999)]{Rudelson99randomvectors}
M.~Rudelson.
\newblock Random vectors in the isotropic position.
\newblock \emph{Journal of Functional Analysis}, pages 60--72, 1999.

\bibitem[Soltanolkotabi et~al.(2013)Soltanolkotabi, Elhamifar, and
  Candes]{robust_spc}
Mahdi Soltanolkotabi, Ehsan Elhamifar, and Emmanuel Candes.
\newblock Robust subspace clustering.
\newblock \emph{arXiv:1301.2603}, 2013.

\bibitem[Srebro and Jaakkola(2005)]{Srebro05generalizationerror}
Nathan Srebro and Tommi Jaakkola.
\newblock Generalization error bounds for collaborative prediction with
  low-rank matrices.
\newblock In \emph{Neural Information Processing Systems}, pages 5--27, 2005.

\bibitem[Tron and Vidal(2007)]{hopkin155}
Roberto Tron and Rene Vidal.
\newblock A benchmark for the comparison of 3-d motion segmentation algorithms.
\newblock In \emph{IEEE Conference on Computer Vision and Pattern Recognition},
  pages 1--8, 2007.

\bibitem[Weimer et~al.(2007)Weimer, Karatzoglou, Le, and
  Smola]{nips:WeimerKLS07}
Markus Weimer, Alexandros Karatzoglou, Quoc~V. Le, and Alex~J. Smola.
\newblock Cofi rank - maximum margin matrix factorization for collaborative
  ranking.
\newblock In \emph{Neural Information Processing Systems}, 2007.

\bibitem[Xu et~al.(2010)Xu, Caramanis, and Sanghavi]{xu:2010:nips}
Huan Xu, Constantine Caramanis, and Sujay Sanghavi.
\newblock Robust pca via outlier pursuit.
\newblock In \emph{Neural Information Processing Systems}, 2010.

\bibitem[Xu et~al.(2013)Xu, Caramanis, and Mannor]{xu:2013:tit}
Huan Xu, Constantine Caramanis, and Shie Mannor.
\newblock Outlier-robust pca: The high-dimensional case.
\newblock \emph{IEEE Transactions on Information Theory}, 59\penalty0
  (1):\penalty0 546--572, 2013.

\bibitem[Zhang et~al.(2012)Zhang, Ganesh, Liang, and Ma]{zhang:2012:ijcv}
Zhengdong Zhang, Arvind Ganesh, Xiao Liang, and Yi~Ma.
\newblock Tilt: Transform invariant low-rank textures.
\newblock \emph{International Journal of Computer Vision}, 99\penalty0
  (1):\penalty0 1--24, 2012.

\end{thebibliography}
\end{document}